\documentclass[aps,pra, twocolumn]{revtex4-1}
\usepackage{amssymb,amsthm,amsmath,enumerate,url, natbib, graphicx, float, csvsimple, todonotes, enumitem}

\newcommand{\bs}[1]{\boldsymbol{#1}}

\newcommand{\M}{\mathcal{M}}

\newcommand{\Tr}{\mathrm{Tr}}

\newcommand{\reals}{\mathbb{R}}
\newcommand{\cH}{\mathcal{H}}
\newcommand{\cL}{\mathcal{L}}
\newcommand{\Id}{\mathbb{I}}

\newcommand{\ket}[1]{\ensuremath{\left|#1\right\rangle}}

\newcommand{\expect}[1]{\ensuremath{\left\langle#1\right\rangle}}

\def\Id{1\!\mathrm{l}}
\newcommand{\Fi}{\mathcal{I}}

\newcommand{\bvec}[1]{\boldsymbol{#1}}

\newcommand{\rhohat}{\hat{\rho}}

\newcommand{\rhoML}[1]{\rhohat_{\scriptscriptstyle{\mathrm{ML},#1}}}

\newtheorem{lem}{Lemma}
\newtheorem{mydef}{Definition}

\begin{document}
\author{Travis L Scholten}
\author{Robin Blume-Kohout}
\affiliation{Center for Computing Research (CCR), Sandia National Laboratories}
\affiliation{Center for Quantum Information and Control (CQuIC), University of New Mexico}

\title{Behavior of the Maximum Likelihood in Quantum State Tomography}

\begin{abstract}
Quantum state tomography on a $d$-dimensional system demands resources that grow rapidly with $d$. They may be reduced by using model selection to tailor the number of parameters in the model (i.e., the size of the density matrix).  Most model selection methods typically rely on a test statistic and a null theory that describes its behavior when two models are equally good. Here, we consider the loglikelihood ratio.  Because of the positivity constraint $\rho \geq 0$, quantum state space does not generally satisfy local asymptotic normality, meaning the classical null theory for the loglikelihood ratio (the Wilks theorem) should not be used.  Thus, understanding and quantifying how positivity affects the null behavior of this test statistic is necessary for its use in model selection for state tomography.  We define a new generalization of local asymptotic normality, metric-projected local asymptotic normality, show that quantum state space satisfies it, and derive a replacement for the Wilks theorem. In addition to enabling reliable model selection, our results shed more light on the qualitative effects of the positivity constraint on state tomography.
\end{abstract}
\date{\today}

\maketitle

Determining the quantum state $\rho_{0}$ produced by a specific preparation procedure for a quantum system is a problem almost as old as quantum mechanics itself \cite{Corbett2006, Pauli1933}. This task, known as \emph{quantum state tomography} \cite{Paris2004}, is not only useful in its own right (diagnosing and detecting errors in state preparation), but is also used in other characterization protocols including entanglement verification \cite{Steffen2006, Blume-Kohout2010c, VanEnk2007} and process tomography \cite{Anis2012}. A typical state tomography protocol proceeds as follows: many copies of $\rho_{0}$ are produced, they are measured in diverse ways, and finally the outcomes of those measurements (data) are collated and analyzed to produce an estimate $\rhohat$.  This is a straightforward statistical inference process \cite{Reid2015, Wasserman2004}, where the data are used to fit the parameters of a statistical model. In state tomography, the parameter is $\rho$, and the model is the set of all possible density matrices on a Hilbert space $\cH$ (equipped with the Born rule). However, we don't always know what model to use. It is not always \emph{a priori} obvious what $\cH$ or its dimension is; examples include optical modes \cite{Altepeter2005, Bertrand1987, Lvovsky2009, Breitenbach1997, Leonhardt1995} and leakage levels in AMO and superconducting \cite{Motzoi2009, Fazio1999} qubits. In such situations, we seek to let the data itself determine which of many candidate Hilbert spaces is best suited for reconstructing $\rho_{0}$.

Choosing an appropriate Hilbert space on the fly is an instance of a general statistical problem called \emph{model selection}.  Although model selection has been thoroughly explored in classical statistics \cite{Burnham2004}, its application to state tomography encounters some obstacles.  They stem from the fact that quantum states -- and therefore, estimates of them -- must satisfy a \emph{positivity constraint} $\rho\geq0$.  (See Figure \ref{fig:boundaries}.) A similar constraint, complete positivity, applies to process tomography.  The impact of positivity constraints on state and process tomography is an active area of research \cite{Candes2006, Flammia2012a, Suess2016, Carpentier2015}, and its implications for model selection have also been considered \cite{Schwarz2013a, Guta2012a, VanEnk2013a, Langford2013, Yin2011, Moroder2013, Knips2015}.  In this paper, we address a specific question at the heart of this matter:  \emph{How does the loglikelihood ratio statistic used in many model selection protocols, including (but not limited to) information criteria such as Akaike's AIC \cite{Akaike1974}, behave in the presence of the positivity constraint $\rho\geq0$}?

\begin{figure}
\includegraphics[width=\columnwidth]{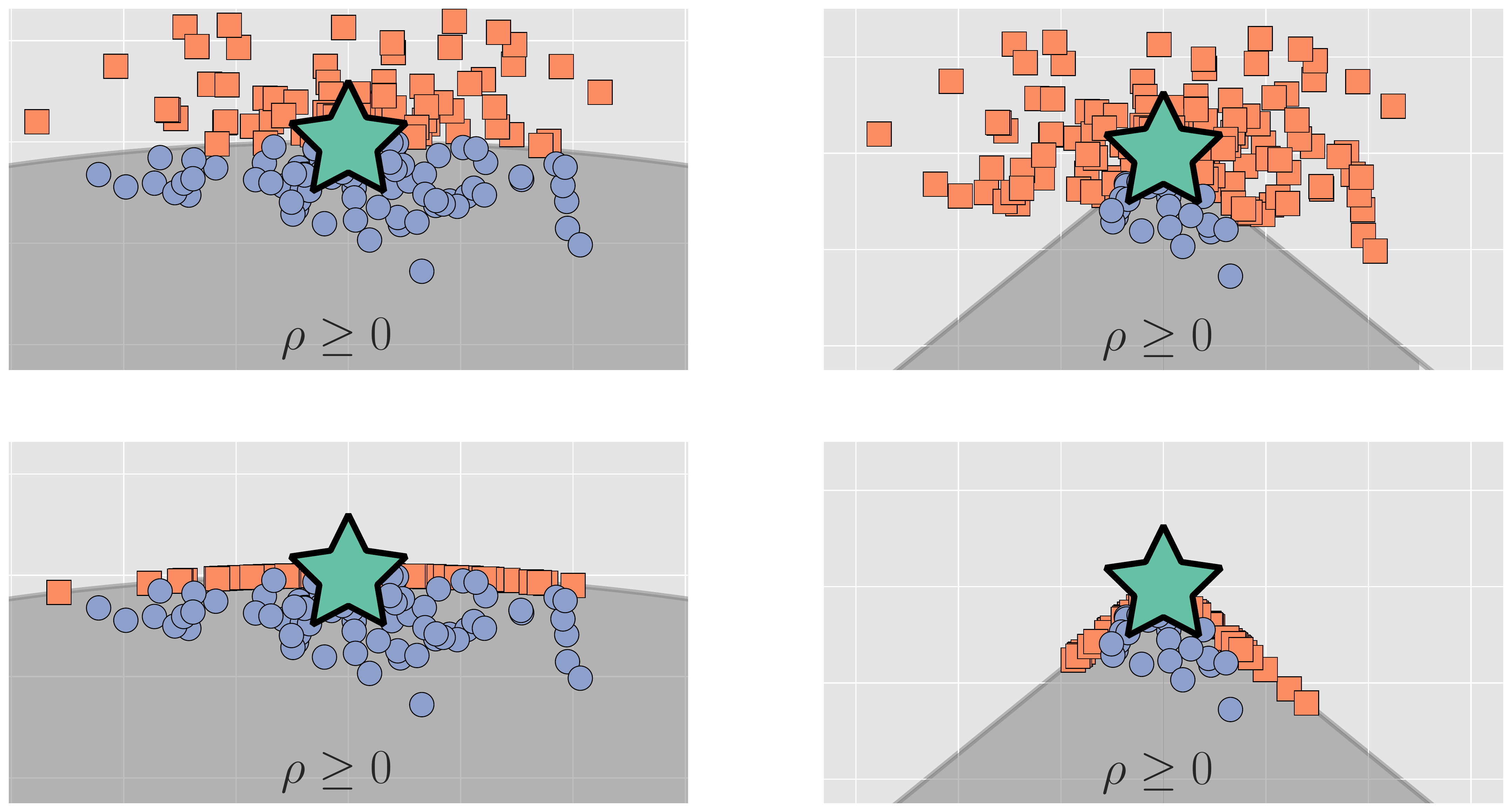}
 \caption{\textbf{Impact of the positivity constraint ($\rho\geq0$)  on tomographic estimates}:  This figure illustrates how the quantum state space's boundary -- which results from the constraint $\rho\geq0$ -- affects maximum likelihood (ML) tomography for a qutrit state $\rho_{0}$ (star).  Two different 2-dimensional cross-sections of the state space are shown, which correspond to a qubit (left) and a classical 3-outcome distribution (right). \textbf{Top}: Without the positivity constraint, some ML estimates (orange squares) are not valid estimates of a quantum state, because they are not positive semidefinite. However, some ML estimates (blue circles) are. Further, the ML estimates are Gaussian distributed.
\textbf{Bottom}:  Imposing the positivity constraint forces the (previously negative) ML estimates to ``pile up" on the boundary of state space; the distribution $\mathrm{Pr}(\hat{\rho}_{\mathrm{ML}})$ is not Gaussian, and local asymptotic normality is not satisfied. In turn, the assumptions necessary to invoke the Wilks theorem are not satisfied either.}
\label{fig:boundaries}
\end{figure}

We begin in Section \ref{sec:intro} by introducing the loglikelihood ratio statistic $\lambda$, and outline how it can be used to choose a Hilbert space.  In Section \ref{sec:qstmodelselection}, we show how and why the classical null theory for its behavior, the Wilks theorem, falls apart in the presence of the positivity constraint, because quantum state space does not generally satisfy \emph{local asymptotic normality} (LAN).  We define a new generalization of LAN, \emph{metric-projected local asymptotic normality} (MP-LAN), in Section \ref{sec:lanreplacement}; this generalization explicitly accounts for the positivity constraint, and is satisfied by quantum state space. Using this generalization, we derive a closed-form approximation for $\lambda$'s expected value in Section \ref{sec:computingllrs}, thereby providing a replacement for the Wilks theorem that \emph{is} applicable in state tomography.  Finally, we test the validity of our approximate theory under the assumptions used in its derivation (Section \ref{sec:theorycomp1}), and under harsh conditions by comparing it to numerical results for the realistic problem of optical heterodyne tomography (Section \ref{sec:heterotomo}).

\section{Introduction - Statistical Model Selection}
\label{sec:intro}
Discussing model selection for state tomography requires introducing some terminology/notation from statistics.  A \emph{model} is a parameterized family of probability distributions over some data $D$, usually denoted as $\mathrm{Pr}_{\bs{\theta}}(D)$, where $\bs{\theta}$ are the \emph{parameters} of the model. In state tomography, the parameters are a quantum state $\rho$, the data are the observed outcomes of the measurement of a positive operator-valued measure (POVM) $\{E_{j}\}$, and the probability of observing outcome ``$j$" is given by the Born rule: $p_{j} = \mathrm{Tr}(\rho E_{j})$ \footnote{The index $j$ may be continuous or discrete.}. In this paper, a model is a set of density matrices, and a state $\rho$ is a particular choice of the model's parameters.

Suppose we have data $D$ obtained from an unknown state $\rho_{0}$, and two candidate models $\M_{1}, \M_{2}$ that could be used to reconstruct it.  Many of the known methods for choosing between them (i.e., model selection) involve quantifying how well each model fits the data by its \emph{likelihood}.  The likelihood of a simple hypothesis $\rho$ is defined as $\mathcal{L}(\rho) = \mathrm{Pr}(D|\rho)$.  Models, however, are \emph{composite} hypotheses, comprising many possible values of $\rho$.  A canonical way to define model $\M$'s likelihood is via the general method of \emph{maximum likelihood} (ML), by maximizing $\cL(\rho)$ over $\rho\in\M$.  In practice, the maximization is usually done explicitly to find an ML estimate $\hat{\rho}_{\mathrm{ML},\M}$ \cite{Hradil1997, JamesPRA2001, Blume-Kohout2010} of $\M$'s parameters, and then $\cL(\M) = \cL(\hat{\rho}_{\mathrm{ML},\M})$.  (Although it is common to refer to $\hat\rho_{\mathrm{ML}}$ without specifying the model over which $\cL$ was maximized, we list the model explicitly because in this paper, we are frequently using many different models!)

Then, the models can be compared using the \emph{loglikelihood ratio statistic} \cite{Neyman1933, Blume-Kohout2010, Moroder2013}:
\begin{align}
\nonumber \lambda(\M_{1}, \M_{2}) &\equiv -2 \log \left(\frac{\cL(\M_{1})}{\cL(\M_{2})}\right)\\
\nonumber &= -2 \log \left(\frac{\cL(\rhoML{\M_{1}})}{\cL(\rhoML{\M_{2}})}\right)\\
&= -2 \log \left(\frac{\underset{\rho \in \M_{1}}{\max}~\cL(\rho)}{\underset{\rho \in \M_{2}}{\max}~\cL(\rho)}\right).
\end{align}
All else being equal, a positive $\lambda$ favors $\M_2$ -- i.e., the model with the higher likelihood is more plausible, because it fits the data better.  However, all else is rarely equal.  If both models are equally valid -- e.g., they both contain $\rho_0$ -- but $\M_2$ has more parameters, then $\M_2$ will very probably fit the data better.  Models with more adjustable parameters do a better job of fitting \emph{noise} (e.g., finite sample fluctuations) in the data.  This becomes strictly true when the models are \emph{nested}, so that $\M_{1} \subset \M_{2}$.  In this case, the likelihood of $\M_{2}$ is at least as high as that of $\M_{1}$;  not only is $\lambda \geq 0$, but almost surely $\lambda > 0$.

Remarkably, the same effect also makes $\M_{2}$'s fit less accurate (almost surely), because the fit incorporates more of the noise in the data.  These two effects constitute \emph{overfitting}, which can be summed up as ``Extra parameters make the fit look better, but perform worse.''.  An overfitted model would fit \emph{current} data extremely well,  but would fail to accurately predict \emph{future} data. This provides strong motivation to correct for overfitting by penalizing or handicapping larger models, to prevent them from being chosen over smaller models that are no less valid, and may even yield better estimates in practice \cite{Akaike1974}.

For this reason, any model selection method/criterion that relies (explicitly or implicitly) on a statistic to quantify ``how well model $\M$ fits the data'' also relies on a \emph{null theory} to predict how that statistic will behave if some \emph{null hypothesis} is true. For the model selection problems we consider, the null hypothesis is that $\rho_{0} \in \M$, and the null theory will tell us how statistics of interest behave when that null hypothesis is in fact true.  A model selection criterion based on an invalid null theory (or a criterion used in a context where its null theory does not apply) will tend to perform sub-optimally (as compared to a method based on a correct null theory).

The null theory can be used to formulate a \emph{decision rule} for choosing between models. If we know how the test statistic behaves when both models are equally valid, then we can evaluate the observed value of the statistic under the null theory. If the observed value is very improbable under the null theory, then that constitutes evidence against the smaller model, and justifies rejecting it. On the other hand, if the observed value is \emph{consistent} with the null theory, there is no reason to reject the smaller model.

The standard null theory for $\lambda$ is the \emph{Wilks theorem} \cite{Wilks1938}. It relies on \emph{local asymptotic normality} (LAN) \cite{LeCam1970, LeCam1956}. LAN is a property of $\M$; if $\M$ satisfies LAN, then as $N_{\mathrm{samples}}\rightarrow \infty$:
\begin{itemize}[nosep]
\item The ML estimate $\rhoML{\M}$ is normally distributed around $\rho_{0}$ with covariance matrix $\Fi^{-1}$:
\begin{equation}
\label{eq:landist}
\mathrm{Pr}(\rhoML{\M}) \propto \exp\left[-\mathrm{Tr}[(\rho_{0} - \rhoML{\M})\mathcal{I}(\rho_{0} -\rhoML{\M})]/2\right].
\end{equation}
\item The likelihood function in a neighborhood of $\rhoML{\M}$ is locally Gaussian with Hessian $\Fi$:
\begin{equation}
\label{eq:lanl}
\mathcal{L}(\rho) \propto \exp\left[-\mathrm{Tr}[(\rho - \rhoML{\M})\mathcal{I}(\rho - \rhoML{\M})]/2\right].
\end{equation}
\end{itemize}
Here, $\Fi$ is the (classical) \emph{Fisher information matrix} associated with the POVM. It quantifies how much information the data carry about a parameter in the model.  (Note that in expressions involving $\mathcal{I}$, states $\rho$ are treated as vectors in state space, and $\mathcal{I}$ is a matrix or 2-index tensor acting on that state space.)

Most statistical models satisfy LAN.  When LAN is satisfied \emph{and} $N_{\mathrm{samples}}$ is large enough to reach the ``asymptotic" regime, we can invoke the Wilks theorem to determine the behavior of $\lambda$. This theorem says that under suitable regularity conditions, if $\rho_{0}\in \M_{1}\subset \M_{2}$, where $\M_{2}$ has $K$ more parameters than $\M_{1}$, then $\lambda$ is a $\chi^{2}_{K}$ random variable.  This is a complete null theory for $\lambda$ (under the specified conditions), and implies that $\langle \lambda \rangle = K$ and $(\Delta \lambda)^{2} = 2K$.

Therefore, in the ``Wilks regime", a simple criterion for model selection would be to compare the observed value of $\lambda$ to $\lambda_{\mathrm{thresh}} = \langle \lambda \rangle + k\Delta \lambda$, for some $k \approx 1$, and reject the smaller model if $\lambda > \lambda_{\mathrm{thresh}}$.  While model selection rules can be more subtle and complex than this \cite{Akaike1974, Schwarz1978, Kass1995, Spiegelhalter2002}, they usually take the general form of a threshold in which $\expect{\lambda}$ plays a key role.  Rather than attempting to define a specific rule, our purpose in this paper is to understand the behavior of $\expect{\lambda}$ and derive an approximate expression for it in the context of state tomography.

The first step in doing so is to explain how and why the Wilks theorem breaks down in that context.

\section{Quantum State Tomography and The Breakdown of the Wilks Theorem}
\label{sec:qstmodelselection}
Quantum state tomography typically begins with $N_{\mathrm{samples}}$ independently and identically prepared quantum systems -- i.e., $N_{\mathrm{samples}}$ copies of an unknown state $\rho_{0}$.  Each copy is measured, and without loss of generality we can assume that each measurement is described by the same positive operator-valued measure (POVM).  A POVM is a collection of positive operators $\{E_j\}$ summing to $\Id$, and the probability of outcome ``$j$'' is given by $\Tr(\rho_0 E_j)$.  The results of all $N_{\mathrm{samples}}$ measurements constitute data, represented as a record of the frequencies of the possible outcomes $\{n_{j}\}$, where $n_{j}$ is the number of times ``$j$'' was observed, and $\sum_{j}n_{j} = N_{\mathrm{samples}}$.  Finally, this data is processed through some \emph{estimator} to yield an estimate of $\rho_0$, denoted $\hat{\rho}$ . 

Although a variety of estimators have been proposed \cite{Vogel1989,Hradil1997,JamesPRA2001,Blume-Kohout2010b,Blume-Kohout2010,Zhu2014a,Ferrie2016}, the exact estimator used is not our concern here.  However, since we \emph{are} concerned with computing the likelihood of a model $\M$, which is defined as the likelihood of the most likely $\rho\in\M$, we will make extensive use of the \emph{maximum likelihood} (ML) estimator.  This should not be taken as advocacy for the ML estimator; it is only a convenient way to find the maximum of $\cL$ over $\M$, and once a model is chosen, a different estimator could be used.
The likelihood $\mathcal{L}(\rho)$ is
\begin{equation}
\nonumber \mathcal{L}(\rho) = \prod_{j}\mathrm{Tr}(\rho E_{j})^{n_{j}},
\end{equation}
and $\rhoML{\M}$ is the solution to the optimization problem
\begin{equation}
\nonumber \rhoML{\M} = \underset{\rho \in \M}{\text{argmax}}~\mathcal{L}(\rho).
\end{equation}
In state tomography, $\M$ is almost always the set of all density matrices over a Hilbert space $\cH$:
\begin{equation}
\nonumber \mathcal{M}_{\cH} = \{\rho~|~\rho \in \mathcal{B}(\mathcal{H}),~\mathrm{Tr}(\rho) =1,~\rho \geq 0\},
\end{equation}
where $\mathcal{B}(\cH)$ is the space of bounded linear operators on $\cH$.  To determine $\rhoML{\M}$, we can use the following facts: (a) $\M_{\cH}$ is a convex set, and (b) $\rhoML{\M}$ minimizes the value of the convex function $-\log[\mathcal{L}(\rho)]$. Because $\rhoML{\M}$ is the solution to minimizing a convex function  over a convex set, it can be found efficiently via any of several algorithms for convex optimization \cite{Boyd}.

Usually, $\cH$ is taken for granted or chosen by fiat.  In this paper, we will consider a nested family of different Hilbert spaces, indexed by their dimension $d$: $\cH_{1}  \subset \cdots \subset \cH_{d} \subset \cH_{d+1} \subset \cdots$.  The models we consider are therefore given by:
\begin{equation}
\label{eq:modelsd}
\M_{d} \equiv \mathcal{M}_{\cH_{d}} = \{\rho~|~\rho \in \mathcal{B}(\mathcal{H}_{d}),~\mathrm{Tr}(\rho) =1,~\rho \geq 0\}.
\end{equation}
For notational brevity, we will use $\rhoML{d}$ to denote the ML estimate over $\M_{d}$. To select between these models, we need to determine whether one model (say, $\M_{d + 1}$) is ``better'' than another (say, $\M_{d}$).  To evaluate which is better, we typically compute the likelihood of each model, and then use $\lambda(\M_{d}, \M_{d+1})$ to choose between them. As mentioned in the previous section, this requires having a \emph{null theory} for $\lambda$ that describes its behavior when $\rho_{0} \in \M_{d} \subset \M_{d + 1}$.

The Wilks theorem, which is the classical null theory for $\lambda$, relies on local asymptotic normality (LAN). If the models under consideration satisfy LAN, then as mentioned in the previous section, the likelihood $\mathcal{L}(\rho)$ is Gaussian with a Hessian given by the Fisher information. In classical statistics, it is common to assume that boundaries are not relevant, either because the models of interest have none, or because the true parameter values $\rho_{0}$ lie far away from them.  In the absence of boundaries, and in the asymptotic limit where the curvature of the Fisher information metric is also negligible, many calculations can be simplified by changing to \emph{Fisher-adjusted} coordinates in which the Fisher information is isotropic (i.e., $\mathcal{I}\propto\Id$). Under these assumptions and simplifications, the Wilks theorem can be derived.

In quantum state tomography, the Wilks theorem breaks down for two reasons. First, the quantum state space \emph{does} have boundaries.  Second, the Fisher information is anisotropic, and the anisotropy can't easily be eliminated by a coordinate change because those boundaries define a preferred coordinate system. We discuss these obstacles -- and our plan to address them -- in detail in the remainder of this section.

Given a model $\M_{d}$, its boundary is the set of rank-deficient states within it. When $\rho_{0}\in \M_{d}$ and is full rank, LAN will hold -- which is to say that, asymptotically, the boundary can indeed be ignored. But when $\rho_{0}$ is rank-deficient, it lies \emph{on} the boundary of the model.  LAN is not satisfied, because positivity constrains $\rhoML{d}$, and so $\mathrm{Pr}(\rhoML{d})$ is not Gaussian (see Figure \ref{fig:boundaries}). The Wilks theorem does not apply in this case, and its predictions regarding $\langle \lambda \rangle$ aren't even close (see Figure \ref{fig:boundaries2}). Moreover, this is the relevant situation for our analysis, because \emph{even if $\rho_{0}$ is full-rank in $\M_{d}$, it must be rank-deficient in $\M_{d+1}$}. So we require a replacement for the Wilks theorem; that is, we need a null theory for $\lambda$ when $\rho_0$ is rank-deficient.

One challenge in deriving this replacement is that the Fisher information generally depends strongly on $\rho_{0}$ and the POVM being measured (see Figure \ref{fig:anisofi}).  In many standard derivations, such anisotropy has no impact and can be eliminated easily by changing to Fisher-adjusted coordinates.  But the models we consider (quantum states) have boundaries that break scale-invariance, and define preferred coordinate systems. Changing to Fisher-adjusted coordinates does not eliminate the effect of anisotropy, because the boundary has a new shape in the new coordinates that serves as a record of the anisotropy.  Moreover, the methods we derive here for calculating the impact of the boundary rely heavily on a particular coordinate system (Hilbert-Schmidt coordinates), and changing to Fisher-adjusted coordinates would break them.  This makes it very difficult to derive an precise generalization of the Wilks theorem for \emph{arbitrary} Fisher information, so to derive our results we make the key simplifying assumption that the Fisher information is isotropic with respect to Hilbert-Schmidt metric.  This is almost never exactly true in practice \footnote{Jonathan A Gross, private communication}, but it is reasonable to presume that our results remain useful and approximately true when the Fisher information is \emph{almost} isotropic.  Our results actually appear to be surprisingly robust to significant anisotropy.  In Section \ref{sec:heterotomo}, we perform numerical simulations of heterodyne tomography -- in which the condition number of the Fisher information ranges from $10^1$ to $10^9$ (1 corresponds to isotropic Fisher information) -- and find that our new theory remains reasonably accurate even in this extreme scenario.

To derive our replacement for the Wilks theorem, we first need a new framework for reasoning about models with convex constraints. We develop such a framework in the next section by defining a new generalization of LAN.

\begin{figure}
\includegraphics[width=\columnwidth]{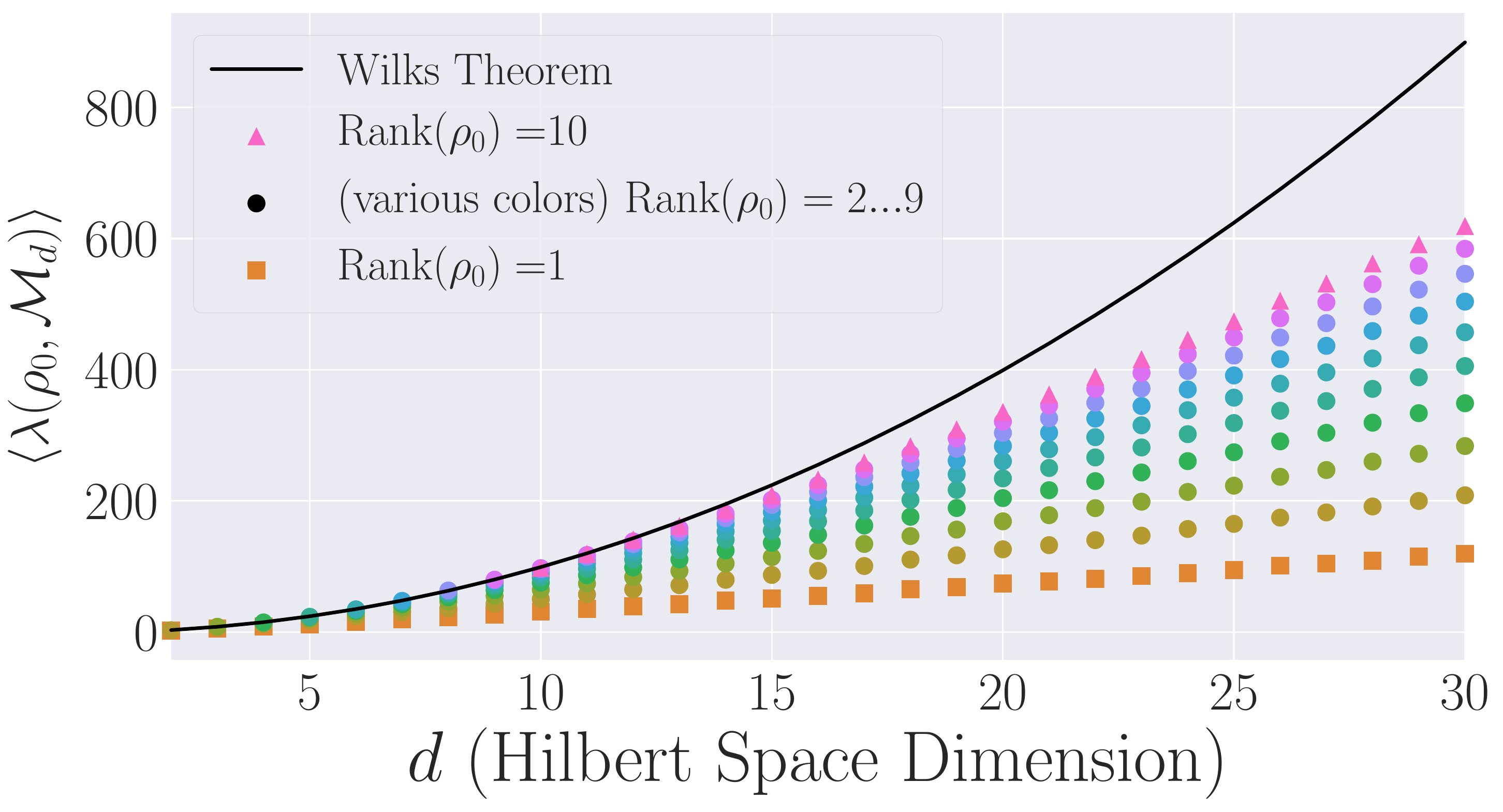}
 \caption{\textbf{Predictions of the Wilks theorem vs reality:}  In the context of state tomography on a true state $\rho_0$ in a $d$-dimensional Hilbert space, the Wilks theorem can be used to predict that, when comparing the zero-parameter model $\M_{0} = \{\rho_0\}$ and the $(d^2-1)$-parameter model $\mathcal{M}_d$ defined in Equation \eqref{eq:modelsd}, the expected loglikelihood ratio $\langle \lambda(\M_{0}, \M_{d})\rangle$ will be $d^2-1$.  Here, we compare that prediction to numerical simulations of tomography on states $\rho_0$ in dimension $d=2,\ldots,30$, with ranks $r=1,\ldots,\mathrm{min}(10,d)$.  The Wilks theorem only predicts $\expect{\lambda}$ correctly for full-rank states; when $r \ll d$, the actual expected loglikelihood ratio is much smaller. Our main result (Equation \ref{eq:ourLLRS}) gives a replacement that works correctly (see Figure \ref{fig:modelcomp-iso}).}
\label{fig:boundaries2}
\end{figure}

\section{Generalizing LAN to Deal with Constraints}
\label{sec:lanreplacement}

In this section, we develop a framework that will allow us to derive a replacement for the Wilks theorem that holds for rank-deficient $\rho_{0}$. To do so, we define a generalization of LAN in the presence of boundaries, which we call \emph{metric-projected local asymptotic normality} (MP-LAN). (For other generalizations of LAN, see \cite{Roussas2010, Jeganathan1982}.) Like LAN, MP-LAN is a property that a statistical model may satisfy. Unlike LAN, MP-LAN is satisfied by quantum state space. For any model that satisfies MP-LAN (quantum or classical), we compute an asymptotically exact expression for $\lambda$, a necessary building block in our replacement for the Wilks theorem. 

In Section \ref{sec:computingllrs}, we show that the models $\M_{d}$ satisfy MP-LAN, and derive an approximation for $\expect{\lambda}$  (Equation \eqref{eq:ourLLRS}, on  page \pageref{eq:ourLLRS}).  Section \ref{sec:theorycomp1} compares our theory to numerical results, and Section \ref{sec:heterotomo} applies our theory to the problem of heterodyne tomography.

The reader should note that to enhance readability, in this section (and only this section) we use $N$ to denote the number of samples, previously denoted as $N_{\mathrm{samples}}$.

\subsection{Definitions and Overview of Results}
\label{sec:lanoverview}
The main definitions and results required for the remainder of the paper are presented in this subsection. Technical details and proofs are presented in the next subsection.

\begin{mydef}[Metric-projected local asymptotic normality, or MP-LAN]
A  model $\M$ satisfies MP-LAN if $\M$ is a convex subset of a model $\M'$ that satisfies LAN.
\end{mydef}

While there are many possible choices for the unconstrained model $\M'$, we will find it useful to let $\M'$ be a model whose dimension is the same as $\M$, but where any of the constraints that define $\M$ are lifted. (For example, in Lemma \ref{lem:qlan}, we will take $\M'$ to be Hermitian matrices of dimension $d$.) Other choices of $\M'$ are possible, but we do not explore those here.

Although the definition of MP-LAN is rather short, it implies some very useful properties. These properties follow from the fact that, as $N \rightarrow \infty$, the behavior of $\rhoML{\M}$ and $\lambda$ is entirely determined by their behavior in an arbitrarily small region of $\M$ around $\rho_{0}$, which we call the \emph{local state space}.

\begin{mydef}[Local state space]
For each natural number $N$, let $I_N$ be the Fisher information matrix of $\mathcal{M}$ at $\rho_0$, and let $
\mathcal{M}_N$ be the set obtained by re-scaling each point in $\mathcal{M}$ by $I_{N}^{-1/2}$. For each $N$, let
 $C_N$ be a convex subset of $\mathcal{M}_N$, chosen so that (a) $C_{N+1}$ contains $C_N$, and (b) $\lim_{N\rightarrow \infty} \mathrm{Pr}(\rhoML{\mathcal{M}} \in C_N) = 1$. Then the sequence $\{C_N: N=1,2...\}$ converges to the local state space around $\rho_0$.
\end{mydef}

Models that satisfy MP-LAN have the following \emph{asymptotic} properties:
\begin{itemize}[nosep]
\item The local state space is the \emph{solid tangent cone} of the model at $\rho_{0}$, denoted $T(\rho_{0})$.
\item The ML estimate $\rhoML{\M}$ is given by the \emph{metric projection} of $\rhoML{\M'}$ onto $T(\rho_{0})$:
\begin{equation}
\label{eq:MP-LANmle}
\rhoML{\M} = \underset{\rho \in T(\rho_{0})}{\text{argmin}}~(\rho  -\rhoML{\M'})\mathcal{I}(\rho  -\rhoML{\M'}).
\end{equation}
(We first encountered the term ``metric projection" in the convex optimization literature \cite{McCoy2014, Amelunxen2014}, and inspires our choice for the acronym ``MP-LAN". However, it should be noted that in the problem setting considered in those references, $\mathcal{I} = \Id$.)

\item The loglikelihood ratio $\lambda(\rho_{0}, \M)$, defined as
\begin{equation}
\label{eq:llrs_lan_2}
\lambda(\rho_{0}, \M) = -2 \log \left(\frac{\cL(\rho_{0})}{\underset{\rho \in \M}{\max}~\cL(\rho)}\right),
\end{equation}
takes the following simple form:
\begin{equation}
\label{eq:llrs_lan}
\lambda(\rho_{0}, \M) =  \mathrm{Tr}[(\rho_{0} - \rhoML{\M})\mathcal{I}(\rho_{0} - \rhoML{\M})].
\end{equation}
(This property is non-trivial; see Figure \ref{fig:llrs_MP-LAN}.)
\end{itemize}

Even when $\M$ satisfies MP-LAN, these properties may not be true when $N$ is finite; they are guaranteed only in the asymptotic limit. When $N$ is sufficiently large, we can (and will!) use the asymptotic properties above.

The following subsection presents the technical details and definitions necessary to show the above results. The reader may skip it without loss of continuity, and proceed to Section \ref{sec:computingllrs}.

\begin{figure}
\includegraphics[width=.75\columnwidth]{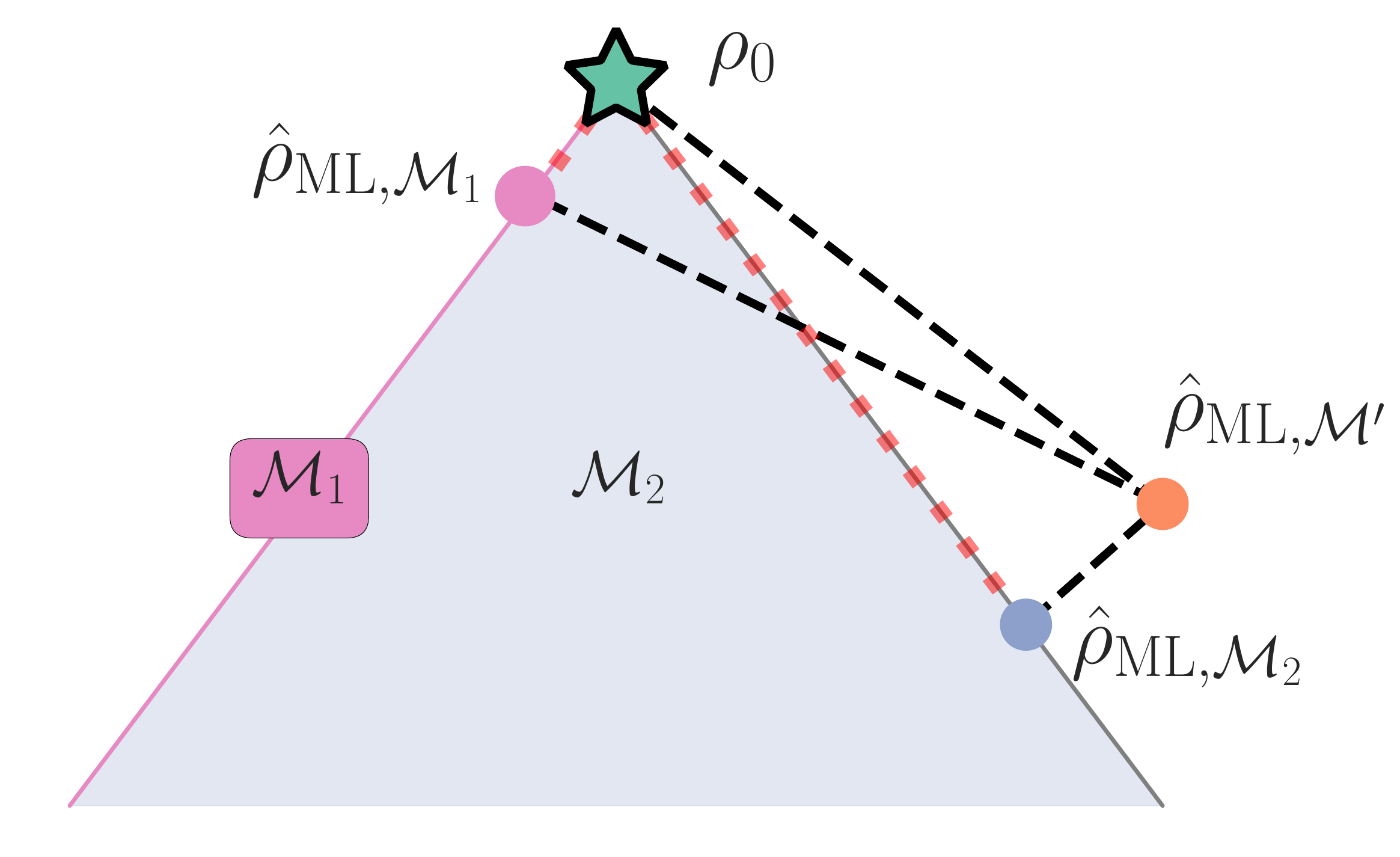}
 \caption{\textbf{Equivalence of $\lambda$ and squared distance when MP-LAN is satisfied:} For any model $\M_{k}$, $\lambda(\rho_{0}, \M_{k})$ is the difference between the squared distance from $\rho_{0}$ to $\rhoML{\M'}$ and that from $\rhoML{\M_{k}}$ to $\rhoML{\M'}$ (black lines). If $\M_{k}$ satisfies MP-LAN, then (a) $\M_{k} \subset \M'$ for an \emph{unconstrained model} $\M'$, and (b) $\lambda$ is equal, in the asymptotic limit, to the squared distance from $\rhoML{\M_{k}}$ to $\rho_{0}$ (red lines), because $\rho_{0}, \rhoML{\M'}$, and $\rhoML{\M_{k}}$ form a right triangle. This is also true for models with \emph{curved} boundaries (such as quantum state space) because asymptotically, the local state space is the solid tangent cone, whose boundaries are always \emph{flat}.}
\label{fig:llrs_MP-LAN}
\end{figure}

\subsection{Technical Details}

Assume a statistical model $\M$ that satisfies MP-LAN. Below, we prove the properties of $\M$ asserted in Section \ref{sec:lanoverview}. 

\subsubsection{Convergence of the Local State Space to the Solid Tangent Cone}

Because $\M$ satisfies MP-LAN, there exists a model $\M' \supset \M$ of dimension $d'$ that satisfies LAN. This means that as $N \rightarrow \infty$, the distribution of $\rhoML{\M'}$ converges to a Gaussian:
\[\mathrm{Pr}(\rhoML{\M'})\xrightarrow[]{\text{d}} \mathcal{N}(\rho_{0}, \Sigma/N),\]
where $\xrightarrow[]{\text{d}}$ means ``converges in distribution to", and $\Sigma = \mathcal{I}^{-1}$.
The shape of the distribution is entirely determined by $\mathcal{I}$. As $N \rightarrow \infty$, this Gaussian distribution becomes more and more tightly concentrated around $\rho_{0}$. Although there is always a non-zero probability that $\rhoML{\M'}$ will be arbitrarily far away from $\rho_{0}$, it is possible to define a sequence of balls $B_{N}$ that shrink with $N$, yet contain every $\rhoML{\M'}$ with probability 1 as $N \rightarrow \infty$.

First, we switch coordinates by sending $\rho \rightarrow \rho - \rho_{0}$, establishing $\rho_{0}$ as the origin of the coordinate system. In these coordinates, $\rhoML{\M'} \sim \mathcal{N}(0, \Sigma/N)$, and the following lemma constructs $B_{N}$.

\begin{lem}
\label{eq:lem}
Let $\rhoML{\M'} \sim \mathcal{N}(0, \Sigma/N)$, and let $\lambda_{\max}(\Sigma)$ denote the largest eigenvalue of $\Sigma$. Define $B_{N} = \{\rho \in \M' ~|~\mathrm{Tr}(\rho^{2}) \leq r^{2}\}$, where $r = \sqrt{\lambda_{\max}(\Sigma)}/N^{1/4}$. Then, $\lim_{N \rightarrow \infty} \mathrm{Pr}(\rhoML{\M'} \in B_{N}) =1$.
\end{lem}

\begin{proof}
Let $B^{0}_{N}$ be an ellipsoidal ball defined by $\{\rho \in \M'~|~\mathrm{Tr}(\rho \Sigma^{-1} \rho) \leq 1/N^{1/2}\}$. Change coordinates by defining $\sigma = N^{1/2}\Sigma^{-1/2}\rho$. In these new coordinates $\hat{\sigma}_{\mathrm{ML},\M'} \sim \mathcal{N}(0, \Id_{d'})$, and $B^{0}_{N} = \{\sigma \in \M'~|~\mathrm{Tr}(\sigma^{2}) \leq N^{1/2}\}$. Therefore,
\begin{align*}
\mathrm{Pr}(\hat{\sigma}_{\mathrm{ML},\M'} \in B^{0}_{N}) &= \mathrm{Pr}(\mathrm{Tr}(\hat{\sigma}_{\mathrm{ML},\M'}^{2}) \leq N^{1/2})\\
&= \int_{0}^{N^{1/2}}\chi^{2}_{d'}(z)~dz,
\end{align*}
because $\mathrm{Tr}(\hat{\sigma}_{\mathrm{ML},\M'}^{2})$ is a $\chi^{2}_{d'}$ random variable. It follows that
\[\lim_{N \rightarrow \infty}~\mathrm{Pr}(\hat{\sigma}_{\mathrm{ML},\M'} \in B^{0}_{N})  =\int_{0}^{\infty}\chi^{2}_{d'}(z)~dz =1.\]
Switching back to the original coordinates, we have
\[B^{0}_{N} = \{\rho \in \M'~|~\mathrm{Tr}(\rho \Sigma^{-1} \rho) \leq 1/N^{1/2}\},\]
and $\lim_{N \rightarrow \infty}~\mathrm{Pr}(\rhoML{\M'} \in B^{0}_{N}) = 1$.

Now that we know $B^{0}_{N}$ contains all  $\rhoML{\M'}$ as $N\rightarrow \infty$, we can now show the same holds true for $B_{N}$. It suffices to show $B^{0}_{N} \subset B_{N}$. To see that this is the case, write the equation for $B^{0}_{N}$ in the standard quadratic form for an ellipsoid:
\[B^{0}_{N} = \{\rho \in \M'~|~\mathrm{Tr}(\rho ( N^{1/2}\Sigma^{-1} ) \rho) \leq 1\}.\]
The standard ellipsoid $\{\mathbf{x}~|~\mathbf{x}^{T}A\mathbf{x} \leq 1\}$ has semi-major axes whose lengths $s_{j}$ are related to the eigenvalues $a_{j}$ of $A$: $s_{j} = 1/\sqrt{a_{j}}$. The matrix $A = N^{1/2}\Sigma^{-1}$ has eigenvalues $N^{1/2}/\lambda_{j}$, where $\lambda_{j}$ are the eigenvalues of $\Sigma$. Thus, the lengths of the semi-major axes of $B^{0}_{N}$ are given by $s_{j} = 1/\sqrt{N^{1/2}/\lambda_{j}} = \sqrt{\lambda_{j}}/N^{1/4}$. Letting $\lambda_{\max}(\Sigma)$ denote the largest eigenvalue of $\Sigma$, the longest semi-major axis of $B^{0}_{N}$ has length $\sqrt{\lambda_{\max}(\Sigma)}/N^{1/4}$.
Because $B_{N}$ is a ball whose radius is equal to this length, $B_{N}$ circumscribes $B^{0}_{N}$, and $B^{0}_{N} \subset B_{N}$.

As $B^{0}_{N} \subset B_{N}$, it follows from the monotonicity of probability that $\mathrm{Pr}(\rhoML{\M'} \in B^{0}_{N}) \leq \mathrm{Pr}(\rhoML{\M'} \in B_{N})$. As the asymptotic limit of $\mathrm{Pr}(\rhoML{\M'} \in B^{0}_{N})$ is 1, and $\mathrm{Pr}(\rhoML{\M'} \in B_{N})$ itself is bounded above by 1, it follows from the squeeze theorem that $\lim_{N\rightarrow \infty}\mathrm{Pr}(\rhoML{\M'} \in B_{N})=1$.
\end{proof}

Informally speaking, Lemma \ref{eq:lem} implies that as $N\rightarrow \infty$ ``all the action" about $\rhoML{\M'}$ takes place inside $B_{N}$. Accordingly, to understand the behavior of quantities which depend on $\rhoML{\M'}$ (such as $\rhoML{\M}$ and $\lambda$), it is sufficient to consider their behavior within $B_{N}$. In fact, we can show that, asymptotically, all the $\rhoML{\M}$ are contained within the region $C_{N} \equiv B_{N} \cap \M$:

\begin{lem}
\label{eq:rhoMLlem}
$\lim_{N\rightarrow \infty}\mathrm{Pr}(\rhoML{\M} \in C_{N}) = 1$.
\end{lem}
\begin{proof}
Using the law of total probability, write $\mathrm{Pr}(\rhoML{\M} \in C_{N})$ as a sum of two terms, depending on whether $\rhoML{\M'} \in B_{N}$. Letting $p$ denote $\mathrm{Pr}(\rhoML{\M'} \in B_{N})$, we have
\begin{align*}
\mathrm{Pr}(\rhoML{\M} \in C_{N}) &= p\mathrm{Pr}(\rhoML{\M} \in C_{N} | \rhoML{\M'} \in B_{N})\\
&+(1-p)\mathrm{Pr}(\rhoML{\M} \in C_{N} | \rhoML{\M'} \not\in B_{N})\\
& \geq p\mathrm{Pr}(\rhoML{\M} \in C_{N} | \rhoML{\M'} \in B_{N}).
\end{align*}
For any $\rhoML{\M'} \in B_{N}$, the corresponding $\rhoML{\M}$ is somewhere in $\M$. To show $\rhoML{\M} \in C_{N}$, we use a proof by contradiction. Suppose that $\rhoML{\M}$ is the ML estimate in $\M$ for $\rhoML{\M'}$, and that $\rhoML{\M} \not \in C_{N}$. Let $\rho_{C}$ denote the closest point in $C_{N}$ to $\rhoML{\M}$. Because $B_{N} \supset C_{N}$, it follows that $\rho_{C}$ is closer to $\rhoML{\M'}$ than $\rhoML{\M}$, contradicting the assumption $\rhoML{\M}$ was the ML estimate in $\M$ for $\rhoML{\M'}$. Therefore, if $\rhoML{\M'} \in B_{N}$, it must be the case that $\rhoML{\M} \in C_{N}$.

Consequently, $\mathrm{Pr}(\rhoML{\M} \in C_{N} | \rhoML{\M'} \in B_{N}) = 1$, implying $\mathrm{Pr}(\rhoML{\M} \in C_{N}) \geq \mathrm{Pr}(\rhoML{\M'} \in B_{N})$. Applying the squeeze theorem, plus Lemma \ref{eq:lem}, we conclude $\lim_{N\rightarrow \infty}\mathrm{Pr}(\rhoML{\M} \in C_{N}) = 1$.
\end{proof}

In the original coordinates, both $B_{N}$ and the distribution of $\rhoML{\M'}$ shrink with $N$, but $B_N$ shrinks more slowly.  Suppose, instead, that we switch to $N$-dependent coordinates that shrink with the distribution of $\rhoML{\M'}$.  In these coordinates, $\M$ and $\M'$ grow with $N$, and $B_{N}$ also grows (but more slowly).  This  \emph{homothetic transformation} of $\M$, $\M'$, and $B_{N}$ scales all of them up.  As $N \rightarrow \infty$, $B_{N} \rightarrow \mathbb{R}^{d'}$,  and the local state space is the \emph{solid tangent cone} of $\M$ at $\rho_0$.

\begin{mydef}[Homothetic Transformation] Given a convex set $C$, the homothetic transformation of $C$ with respect to any point $X \in C$, with homothety coefficient $h$, is the set $C_{h}$ defined by
\[C_{h} = \{X + hY~|~\forall~Y\in C, Y \neq X\}.\]
\end{mydef}

\begin{mydef}[Solid Tangent Cone] For each point $X$ in a convex set $C$, let $C_{h}$ be the homothetic transformation of $C$ with respect to $X$, with homothety coefficient $h$. Then, the solid tangent cone $T(X)$ is defined as the following limit:
\[T(X) = \lim_{h\rightarrow \infty} C_{h}.\]
\end{mydef}

Tangent cones are a general feature of convex sets; see \cite{Rockafellar1998}, Chapter 6, Section A for more information about them and their properties.

Let $C_{N} = B_{N} \cap \M$ in Hilbert-Schmidt coordinates. 
We show that, in an $N$-dependent coordinate system, $C_{N}$ converges to the solid tangent cone, and is the local state space.

\begin{lem}Consider the set $C_{N} = B_{N} \cap \M$ in Hilbert-Schmidt coordinates, and define $C'_{N} = \{N^{1/2}\rho~\forall~\rho \in C_{N}\}$. Then:
\begin{itemize}
\item [1)] $\lim_{N\rightarrow \infty}C'_{N}$ is the solid tangent cone at $\rho_{0}$.
\item [2)] $\lim_{N\rightarrow \infty}C'_{N}$ is the local state space.
\end{itemize}
\end{lem}
\begin{proof}~\\
\begin{itemize}
\item [1)] By definition, $C'_{N}$ is a homothetic transformation of $C_{N}$, with homothety coefficient $N^{1/2}$. (The homothetic center is $\rho_{0}$; in these coordinates, it is 0.) By definition, $\lim_{N\rightarrow \infty}C'_{N}$ is the solid tangent cone at $\rho_{0}$.
\item [2)] The original set $C_{N}$ is a convex subset of $\M$, and from Lemma \ref{eq:rhoMLlem}, $\lim_{N\rightarrow \infty}\mathrm{Pr}(\rhoML{\M} \in C_{N}) = 1$.  Further, the coordinate system defined by the mapping $\rho \rightarrow N^{1/2}\rho$ turns the (previously $N$-dependent) Fisher information $\mathcal{I}$ into a constant. Thus, $\lim_{N\rightarrow \infty}C'_{N}$ is the local state space.
\end{itemize}
\end{proof}

Therefore, we have shown that, asymptotically, the local state space around $\rho_{0}$ \emph{is} the solid tangent cone $T(\rho_{0})$.
The geometry of $T(\rho_{0})$ depends strongly on $\rho_{0}$. If $\rho_{0}$ is rank-deficient within $\M$, then $T(\rho_{0})$ is the cone whose faces touch $\M$ at $\rho_{0}$. (See Figure \ref{fig:tangentcone} for a rebit example.) However, if $\rho_{0}$ is full-rank, $T(\rho_{0})$ is $\mathbb{R}^{d^{2}-1}$.

\subsubsection{MLE as Metric Projection}

As $N \rightarrow \infty$, all the $\rhoML{\M'}$ are contained within the ball $B_{N}$, and the local state space is the solid tangent cone. Because $\M'$ satisfies LAN, the likelihood function around each $\rhoML{\M'}$ is Gaussian, meaning the optimization problem defining $\rhoML{\M}$ is given by \begin{equation}
\label{eq:MP-LANmle2}
\rhoML{\M} = \underset{\rho \in T(\rho_{0})}{\text{argmin}}~\mathrm{Tr}[(\rho  -\rhoML{\M'})\mathcal{I}(\rho  -\rhoML{\M'})].
\end{equation}

That is, $\rhoML{\M}$ is the \emph{metric projection} of $\rhoML{\M'}$ onto the tangent cone. See Figure \ref{fig:tangentcone} for a rebit example. (Notice that if $\rhoML{\M'} \in T(\rho_{0})$, then $\rhoML{\M} = \rhoML{\M'}$.) What makes this  nontrivial is the replacement of the original state space $\M$, whose geometry can be arbitrarily complicated, with its tangent cone $T(\rho_{0})$. As shown in the next section, cones can be much simpler and tractable.

\begin{figure}
\includegraphics[width=.75\columnwidth]{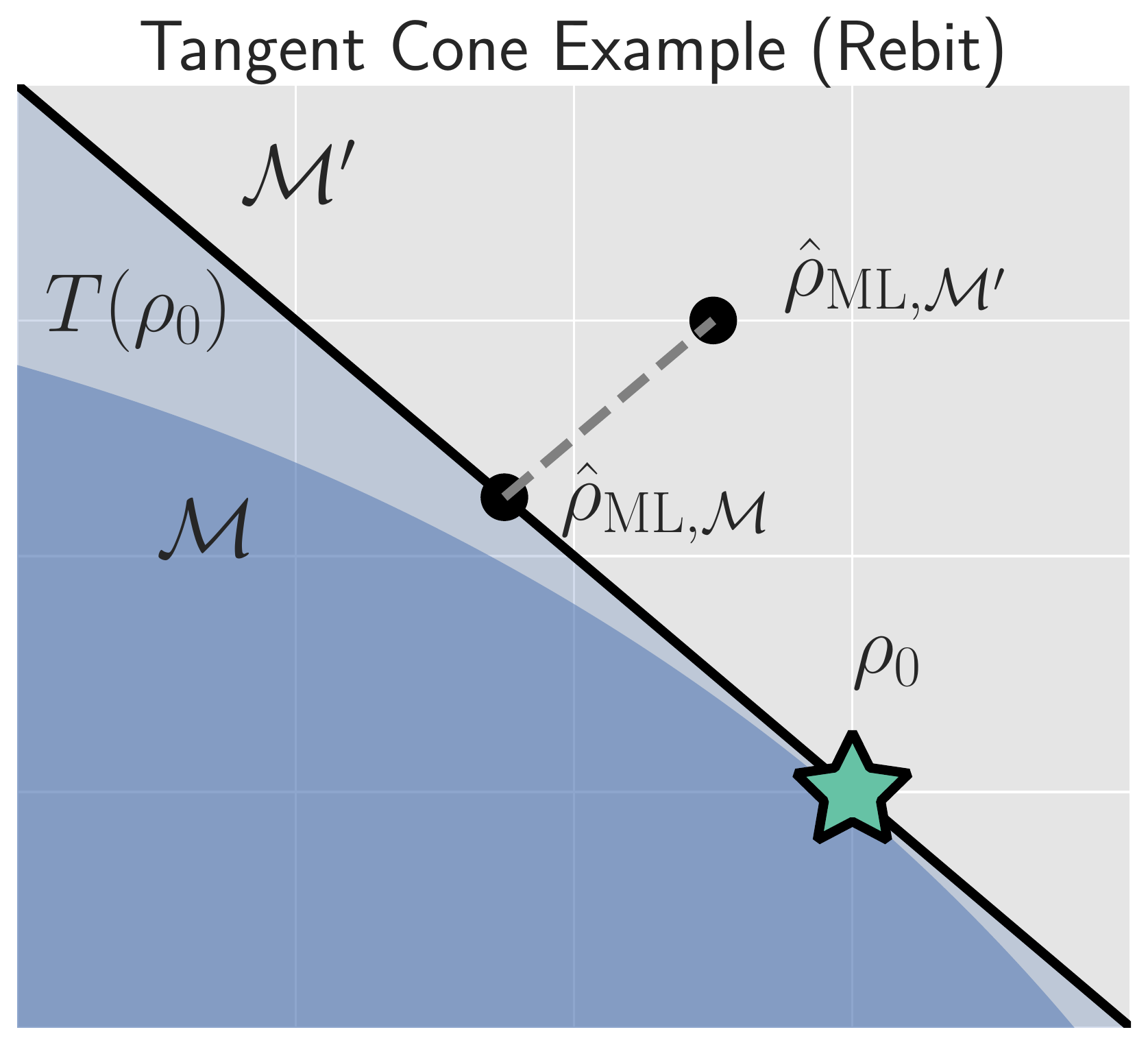}
 \caption{\textbf{Example of the solid tangent cone for a rebit:} As $N \rightarrow \infty$, the local state space around $\rho_{0}$ is $T(\rho_{0})$. In Fisher-adjusted coordinates, it's easy to show that (a) $\rhoML{\M}$ is the metric projection of $\rhoML{\M'}$ onto $T(\rho_{0})$, and (b) $\lambda(\rho_{0}, \M) = \mathrm{Tr}[(\rhoML{\M} - \rho_{0})^{2}]$.}
\label{fig:tangentcone}
\end{figure}

\subsubsection{Expression for $\lambda(\rho_{0}, \M)$}

The loglikelihood ratio statistic between any two models $\lambda(\M_{1}, \M_{2})$ can be computed using a \emph{reference model} $\mathcal{R}$: 
\[\lambda(\M_{1}, \M_{2}) = \lambda(\mathcal{R},\M_{2}) - \lambda(\mathcal{R},\M_{1}),\]
where
\[\lambda(\mathcal{R}, \M) = -2 \log \left(\frac{\cL(\mathcal{R})}{\cL(\M)}\right) =  -2 \log \left(\frac{\underset{\rho \in \mathcal{R}}{\max}~\cL(\rho)}{\underset{\rho \in \M}{\max}~\cL(\rho)}\right).\]
Let us take $\mathcal{R} = \rho_{0}$. Because $\M'$ satisfies LAN, asymptotically $\mathcal{L}(\rho)$ is Gaussian, and $\lambda$ relates to a difference in squared distances:
\begin{align}
\label{eq:lambdalan}
\nonumber \lambda(\rho_{0}, \M)&= -2 \log \left(\frac{\cL(\rho_{0})}{\underset{\rho \in \M}{\max}~\cL(\rho)}\right)\\
\nonumber &\xrightarrow[N \rightarrow \infty]{}~\mathrm{Tr}[(\rho_{0} - \rhoML{\M'})\mathcal{I}(\rho_{0} - \rhoML{\M'})]\\
&-  \mathrm{Tr}[(\rhoML{\M} - \rhoML{\M'})\mathcal{I}(\rhoML{\M} - \rhoML{\M'})].
\end{align}
Using the fact $\rhoML{\M}$ is a metric projection, we can prove that $\lambda(\rho_{0}, \M)$ has a simple form.

\begin{lem}
$\lambda(\rho_{0}, \M) = \mathrm{Tr}[(\rho_{0} - \rhoML{\M})\mathcal{I}(\rho_{0} - \rhoML{\M})]$.
\end{lem}

\begin{proof}
We switch to Fisher-adjusted coordinates ($\rho \rightarrow \mathcal{I}^{1/2}\rho$), and in these coordinates $\mathcal{I}$ becomes $\Id$:
\begin{equation}
\label{eq:lambdalan2}
\lambda(\rho_{0}, \M) = \mathrm{Tr}[(\rho_{0} - \rhoML{\M'})^{2}]-  \mathrm{Tr}[(\rhoML{\M} - \rhoML{\M'})^{2}].
\end{equation}

To prove the lemma, we must consider two cases:

\emph{Case 1}: Assume $\rhoML{\M'} \not \in T(\rho_{0})$. Because $\rhoML{\M}$ is the metric projection of $\rhoML{\M'}$ onto $T(\rho_{0})$ (Equation \eqref{eq:MP-LANmle2}),  the line joining $\rhoML{\M'}$ and $\rhoML{\M}$ is normal to $T(\rho_{0})$ at $\rhoML{\M}$. Because $T(\rho_{0})$ contains $\rho_{0}$ (as its origin), it follows that the lines joining $\rho_{0}$ to $\rhoML{\M}$, and $\rhoML{\M}$ to $\rhoML{\M'}$, are perpendicular. (See Figure \ref{fig:tangentcone}.)

 By the Pythagorean theorem, we have
\[\mathrm{Tr}[(\rho_{0} -\rhoML{\M'})^{2}] =  \mathrm{Tr}[(\rho_{0} - \rhoML{\M})^{2}] + \mathrm{Tr}[(\rhoML{\M} - \rhoML{\M'})^{2}]\]
Subtracting $\mathrm{Tr}[(\rhoML{\M} - \rhoML{\M'})^{2}]$ from both sides, and comparing to Equation \eqref{eq:lambdalan2}, yields the lemma statement in Fisher-adjusted coordinates.

\emph{Case 2}: Assume $\rhoML{\M'} \in T(\rho_{0})$. Then, $\rhoML{\M}= \rhoML{\M'}$, and Equation \eqref{eq:lambdalan2} simplifies to the lemma statement in Fisher-adjusted coordinates.

Switching back from Fisher-adjusted coordinates, we have $\lambda(\rho_{0}, \M) = \mathrm{Tr}[(\rho_{0} - \rhoML{\M})\mathcal{I}(\rho_{0} - \rhoML{\M})]$.
\end{proof}

So if $\M$ satisfies MP-LAN then as $N\rightarrow \infty$ the loglikelihood ratio statistic becomes related to \emph{squared error/loss} (as measured by the Fisher information metric.) This result may be of independent interest in, for example, defining new information criteria, which attempt to balance goodness of fit (as measured by $\lambda$) against error/loss (generally, as measured by squared error).

With these technical results in hand, we can proceed to compute $\langle \lambda(\M_{d}, \M_{d+1})\rangle$ in the next section.

\section{A Wilks theorem for Quantum State Space}
\label{sec:computingllrs}

To derive a replacement for the Wilks theorem, we start by showing  the models $\M_{d}$ satisfy MP-LAN.
\begin{lem}
\label{lem:qlan}
The models $\M_{d}$, defined in Equation \eqref{eq:modelsd}, satisfy MP-LAN.
\end{lem}

\begin{proof} Let $\M'_{d} = \{\sigma ~|~\mathrm{dim}(\sigma) = d, \sigma = \sigma^{\dagger}, \mathrm{Tr}(\sigma)=1\}$. (That is, $\M'_{d}$ is the set of all trace-1, $d \times d$ Hermitian matrices, but we do not require them to be non-negative.) It is clear $\M_{d} \subset \M'_{d}$. Now, $\forall ~\sigma \in \M'_{d}$, the likelihood $\cL(\sigma)$ is twice continuously differentiable, meaning $\M'_{d}$ satisfies LAN. Thus, $\M_{d}$ satisfies MP-LAN.
\end{proof}

We can reduce the problem of computing $\lambda(\M_{d}, \M_{d+1})$ to that of computing $\lambda(\rho_{0}, \M_{k})$ for $k = d, d+1$ using the identity
\[\lambda(\M_{d}, M_{d+1}) = \lambda(\rho_{0}, \M_{d+1}) - \lambda(\rho_{0}, \M_{d}).\]
where $\lambda(\rho_{0}, \M_{k})$ is given in Equation \eqref{eq:llrs_lan_2}.
Because each model satisfies MP-LAN, asymptotically, $\lambda(\rho_{0}, \M_{k})$ takes a very simple form, via Equation \eqref{eq:llrs_lan}:
\[\lambda(\rho_{0}, \M_{k}) = \mathrm{Tr}[(\rho_{0} - \rhoML{\M_{k}})\mathcal{I}_{k}(\rho_{0} - \rhoML{\M_{k}})].\]
The Fisher information $\mathcal{I}_{k}$ is generally anisotropic,  depending on $\rho_{0}$, the POVM being measured, and the model $\M_{k}$ (see Figure \ref{fig:anisofi}). And while the $\rho\geq0$ constraint  that invalidated LAN in the first place is at least somewhat tractable in standard (Hilbert-Schmidt) coordinates, it becomes completely intractable in Fisher-adjusted coordinates.  So, to obtain a semi-analytic null theory for $\lambda$, we will simplify to the case where   $\mathcal{I}_{k} = \Id_{k}/\epsilon^{2} $ for some $\epsilon$ that scales as $1/\sqrt{N_{\mathrm{samples}}}$. (That is, $\mathcal{I}_{k}$ is proportional to the Hilbert-Schmidt metric.) This simplification permits the derivation of analytic results that capture realistic tomographic scenarios surprisingly well \cite{Smolin2012}.

\begin{figure}
\centering
\includegraphics[width=.75\columnwidth]{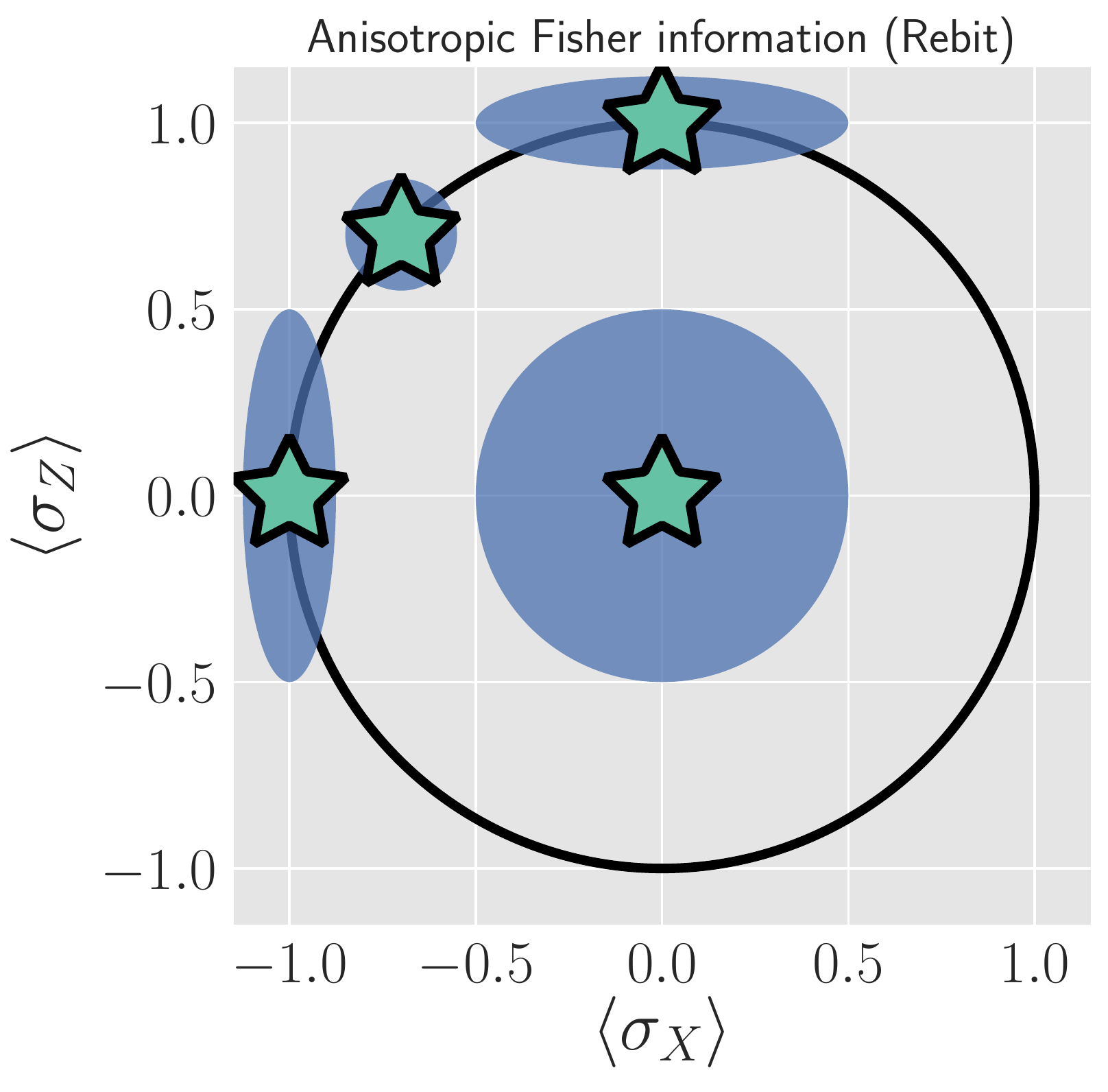}
 \caption{\textbf{Anisotropy of the Fisher information for a rebit:} Suppose a rebit state $\rho_{0}$ (star) is measured using the POVM $\frac{1}{2}\{|0\rangle\langle 0|, |1\rangle\langle 1|, |+\rangle\langle +|, |-\rangle\langle -|\}$. Depending on $\rho_{0}$, the distribution of the unconstrained estimates $\hat{\rho}_{\mathrm{ML}}$ (ellipses) may be anisotropic. Imposing the positivity constraint $\rho \geq 0$ is difficult in Fisher-adjusted coordinates; in this paper, we simplify these complexities to the case where $\mathcal{I} \propto \Id$, \emph{and is independent of $\rho_{0}$}.}
\label{fig:anisofi}
\end{figure}

With this simplification, $\lambda(\M_{d},\M_{d+1})$ is given by
\begin{equation}
\label{eq:llrs_split}
\lambda = \frac{1}{\epsilon^{2}}\left(\mathrm{Tr}[(\rho_{0} - \rhoML{d+1})^{2}] -  \mathrm{Tr}[(\rho_{0} - \rhoML{d})^{2}]\right).
\end{equation}
That is, $\lambda$ is a \emph{difference} in Hilbert-Schmidt distances. This expression makes it clear why a null theory for $\lambda$ is necessary: if $\rho_{0} \in \M_{d},\M_{d+1}$, then $\rhoML{d+1}$ will lie further from $\rho_{0}$ than $\rhoML{d}$ (because there are more parameters that can fit noise in the data). The null theory for $\lambda$ tells us how much extra error will be incurred in using $\M_{d+1}$ to reconstruct $\rho_{0}$ when $\M_{d}$ is just as good.

Describing $\mathrm{Pr}(\lambda)$ is difficult because the distributions of $\rhoML{d}, \rhoML{d+1}$ are complicated, highly non-Gaussian, and singular (estimates ``pile up'' on the various faces of the boundary as shown in Figure \ref{fig:boundaries}).  For this reason, we will not attempt to compute $\mathrm{Pr}(\lambda)$ directly.  Instead, we focus on deriving a good approximation for $\langle \lambda \rangle$.

We consider each of the terms in Equation \eqref{eq:llrs_split} separately and focus on computing $\epsilon^{2}\langle \lambda(\rho_{0}, \M_{d}) \rangle = \langle \mathrm{Tr}[(\rhoML{d}  - \rho_{0})^{2}] \rangle$ for arbitrary $d$.
Doing so involves two main steps:
\begin{itemize}
\item[(1)] Identify which degrees of freedom in $\rhoML{\M'_{d}}$ are, and are not, affected by projection onto the tangent cone $T(\rho_{0})$.
\item[(2)] For each of those categories, evaluate its contribution to the value of $\langle \lambda \rangle$.
\end{itemize}

In Section \ref{subsec:dof}, we identify two types of degrees of freedom in $\rhoML{\M'}$, which we call the ``L" and the ``kite". Section \ref{subsec:L} computes the contribution of degrees of freedom in the ``L", and Section \ref{subsec:kite} computes the contribution from the ``kite". The total expected value is given in Equation \eqref{eq:ourLLRS} in Section \ref{subsec:LLRS}, on page \pageref{eq:ourLLRS}.

\subsection{Separating out Degrees of Freedom in $\rhoML{\M'_{d}}$}
\label{subsec:dof}
We begin by observing that $\lambda(\rho_{0}, \M_{d})$ can be written as a sum over matrix elements,
\begin{align}
\nonumber \lambda &=\epsilon^{-2}\mathrm{Tr}[(\rhoML{d} - \rho_{0})^{2}] = \epsilon^{-2}\sum_{jk}|(\rhoML{d}- \rho_{0} )_{jk}|^{2}\\
\nonumber &= \sum_{jk}\lambda_{jk}~~~~\text{where}~~~~\lambda_{jk} = \epsilon^{-2}|(\rhoML{d} - \rho_{0} )_{jk} |^{2},
\end{align}
and therefore $\langle \lambda \rangle = \sum_{jk}\langle\lambda_{jk}\rangle$.  Each term $\langle \lambda_{jk}\rangle$ quantifies the mean-squared error of a single matrix element of $\rhoML{d}$, and while the Wilks theorem predicts $\expect{\lambda_{jk}}=1$ for all $j,k$, due to positivity constraints, this no longer holds. In particular, the matrix elements of $\rhoML{d}$ now fall into two parts:

\begin{enumerate}[noitemsep]
\item Those for which the positivity constraint \emph{does affect} their behavior.
\item Those for which the positivity constraint \emph{does not affect} their behavior, as they correspond to directions on the surface of the tangent cone $T(\rho_{0})$. (Recall Figure \ref{fig:tangentcone} - as a component of $\rhoML{\M'}$ along $T(\rho_{0})$ changes, the component of $\rhoML{\M}$ changes by the same amount. These elements are unconstrained.)
\end{enumerate}
The latter, which lie in what we call the ``L", comprise all off-diagonal elements on the support of $\rho_0$ \emph{and} between the support and the kernel, while the former, which lie in what we call the ``kite", are all diagonal elements \emph{and} all elements on the kernel (null space) of $\rho_0$.

Performing this division is also supported by numerical simulations (see Figure \ref{fig:L}). Matrix elements in the ``L" appear to contribute $\langle \lambda_{jk}\rangle = 1$, consistent with the Wilks theorem, while those in the ``kite" contribute more (if they are within the support of $\rho_{0}$) or less (if they are in the kernel).  Having performed the division of the matrix elements of $\rhoML{\M'_{d}}$, we observe that $\langle\lambda\rangle = \expect{\lambda_{\mathrm{L}}} + \expect{\lambda_{\mathrm{kite}}}$. Because each $\langle \lambda_{jk}\rangle$ is not necessarily equal to one (as in the Wilks theorem), and because many of them are less than 1, it is clear that their total $\langle \lambda \rangle$ is dramatically lower than the prediction of the Wilks theorem. (Recall Figure \ref{fig:boundaries2}.)

\begin{figure}
\includegraphics[width=\columnwidth]{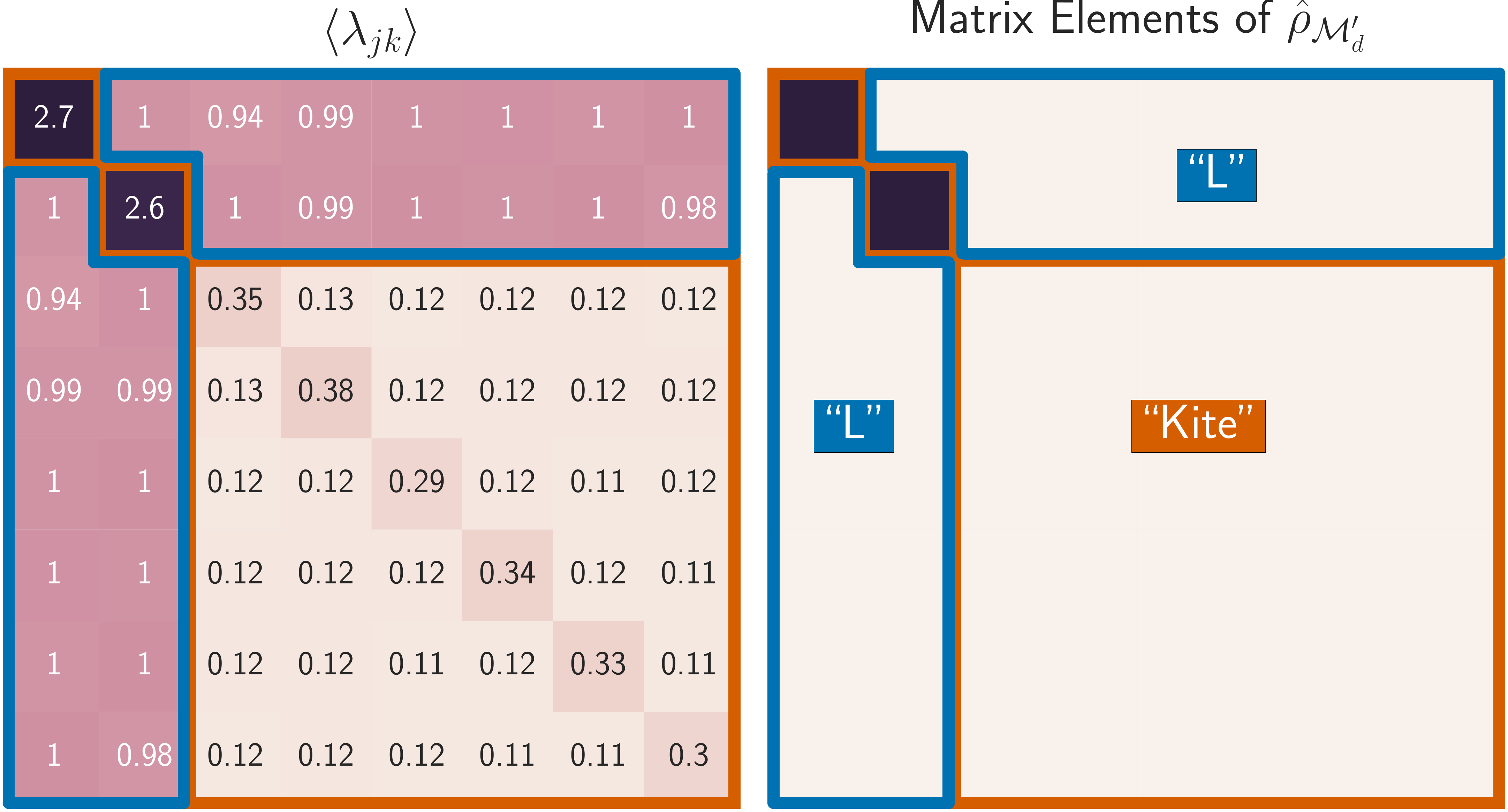}
 \caption{\textbf{Division of the matrix elements of $\rhoML{\M'_{d}}$:} When a rank-2 state is reconstructed in $d=8$ dimensions, the total loglikelihood ratio $\lambda(\rho_0,\mathcal{M}_8)$ is the sum of terms $\lambda_{jk}$ from errors in each matrix element $(\rhoML{d})_{jk}$.  \textbf{Left}:  Numerics show a clear division; some matrix elements have $\expect{\lambda_{jk}}\sim1$ as predicted by the Wilks theorem, while others are either more or less. \textbf{Right}:  The numerical results support our theoretical reasoning for dividing the matrix elements of $\rhoML{\M'_{d}}$ into two parts: the ``kite'' and the ``L''.}
\label{fig:L}
\end{figure}

In the following subsections, we develop a theory to explain the behavior of $\langle \lambda_{\mathrm{L}}\rangle$ and \expect{\lambda_{\mathrm{kite}}}.
In doing so, it is helpful to think about the matrix $\delta \equiv \rhoML{\M'_{d}}- \rho_{0}$, a normally-distributed \emph{traceless} matrix.  To simplify the analysis, we explicitly drop the $\Tr(\delta)=0$ constraint and let $\delta$ be $\mathcal{N}(0,\epsilon^2\Id)$ distributed over the $d^2$-dimensional space of Hermitian matrices (a good approximation when $d\gg2$), which makes $\delta$ proportional to an element of the Gaussian Unitary Ensemble (GUE) \cite{Fyodorov2005}.

\subsection{Computing $\langle \lambda_\mathrm{L}\rangle$}
\label{subsec:L}
The value of each $\delta_{jk}$ in the ``L" is invariant under projection onto the boundary (the surface of the tangent cone $T(\rho_{0})$), meaning that it is also equal to the error $(\rhoML{d} - \rho_{0})_{jk}$. Therefore,  $\langle \lambda_{jk} \rangle = \langle \delta_{jk}^{2}\rangle /\epsilon^{2}$. Because $\M'$ satisfies LAN, it follows that each $\delta_{jk}$ is an i.i.d. Gaussian random variable with mean zero and variance $\epsilon^{2}$. Thus, $\langle \lambda_{jk}\rangle = 1~\forall~(j,k)$ in the ``L". The dimension of the surface of the tangent cone is equal to the dimension of the manifold of rank-$r$ states in a $d$-dimensional space. A direct calculation of that quantity yields $2rd - r(r+1)$, so $\expect{\lambda_{\mathrm{L}}} = 2rd - r(r+1)$.

Another way of obtaining this result is to view the $\delta_{jk}$ in the ``L'' as errors arising due to small unitary perturbations of $\rho_{0}$. Writing $\rhoML{\M'_{d}} = U^{\dagger}\rho_{0}U$, where $U=e^{i\epsilon H}$, we have
\[\rhoML{\M'_{d}} \approx \rho_{0} + i\epsilon [\rho_{0},H]+\mathcal{O}(\epsilon^{2}),\]
and $\delta \approx i\epsilon [\rho_{0},H]$.
If $j = k$, then $\delta_{jj} = 0$. Thus, small unitaries cannot create errors in the diagonal matrix elements, at $\mathcal{O}(\epsilon)$. If $j \neq k$, then $\delta_{jk} \neq 0$, in general. (Small unitaries \emph{can} introduce errors on off-diagonal elements.)

However, if either $j$ or $k$ (or both) lie within the \emph{kernel} of $\rho_{0}$ (i.e., $\langle k | \rho_{0}| k \rangle$ or $\langle j|\rho_{0}|j\rangle$ is 0), then the corresponding $\delta_{jk}$ are zero. The only off-diagonal elements where small unitaries can introduce errors are those which are coherent between the kernel of $\rho_{0}$ and its support. These off-diagonal elements are precisely the ``L", and are  the set $\{\delta_{jk}~|~\langle j | \rho_{0}|j\rangle \neq 0, j\neq k, ~ 0 \leq j,k \leq d - 1\}$. This set contains $2rd - r(r+1)$ elements, each of which has $\langle \lambda_{jk}\rangle = 1$, so we again arrive at $\expect{\lambda_{\mathrm{L}}} = 2rd - r(r+1)$.

\subsection{Computing $\langle \lambda_\mathrm{kite}\rangle$}
\label{subsec:kite}
Computing $\langle \lambda_{\mathrm{L}}\rangle$ was made easy by the fact that the matrix elements of $\delta$ in the ``L" are invariant under the projection of $\rhoML{\M'_{d}}$ onto $T(\rho_{0})$. Computing $\expect{\lambda_{\mathrm{kite}}}$ is a bit harder, because the boundary \emph{does} constrain $\delta$. To understand how the behavior of $\expect{\lambda_{\mathrm{kite}}}$ is affected, we analyze an algorithm presented in \cite{Smolin2012} for explicitly solving the optimization problem in Equation \eqref{eq:MP-LANmle}.

This algorithm, a (very fast) numerical method for computing $\rhoML{d}$ given $\rhoML{\M'_{d}}$, utilizes two steps:
\begin{enumerate}[noitemsep]
\item Subtract $q\Id$ from $\rhoML{\M'_{d}}$, for a particular $q \in \mathbb{R}$.
\item ``Truncate'' $\rhoML{\M'_{d}}-q\Id$, by replacing each of its negative eigenvalues with zero.
\end{enumerate}
Here, $q$ is defined implicitly such that $\Tr\left[ \mathrm{Trunc}(\rhoML{\M'_{d}}-q\Id)\right] = 1$, and must be determined numerically.
However, we can analyze how this algorithm affects the eigenvalues of $\rhoML{d}$, which turn out to be the key quantity necessary for computing $\expect{\lambda_{\mathrm{kite}}}$.

The truncation algorithm above is most naturally performed in the eigenbasis of $\rhoML{\M'_{d}}$.  Exact diagonalization of $\rhoML{\M'_{d}}$ is not feasible analytically, but only its \emph{small} eigenvalues are critical in truncation. Further, only knowledge of the \emph{typical} eigenvalues of $\rhoML{d}$ is necessary for computing $\langle \lambda_{\mathrm{kite}}\rangle$. Therefore, we do not need to determine $\rhoML{d}$ exactly, which would require explicitly solving Equation \eqref{eq:MP-LANmle} using the algorithm presented in \cite{Smolin2012}; instead, we need a procedure for determining its typical eigenvalues.

We assume that $N_{\mathrm{samples}}$ is sufficiently large so that all the nonzero eigenvalues of $\rho_0$ are much larger than $\epsilon$. This means the eigenbasis of $\rhoML{\M'_{d}}$ is accurately approximated by: (1) 
the eigenvectors of $\rho_0$ on its support; and (2) the eigenvectors of $\delta_{\mathrm{ker}} = \Pi_{\mathrm{ker}}\delta\Pi_{\mathrm{ker}} = \Pi_{\mathrm{ker}}\rhoML{\M'_{d}}\Pi_{\mathrm{ker}}$, where $\Pi_{\mathrm{ker}}
$ is the projector onto the kernel of $\rho_0$.

Changing to this basis diagonalizes the ``kite'' portion of $\delta$, and leaves all elements of the ``L'' unchanged (at $\mathcal{O}(\epsilon)$).  The diagonal elements fall into two categories:
\begin{enumerate}[noitemsep]
\item $r$ elements corresponding to the eigenvalues of $\rho_0$, which are given by $p_{j} = \rho_{jj} + \delta_{jj}$ where  $\rho_{jj}$ is the $j^{\mathrm{th}}$ eigenvalue of $\rho_{0}$, and $\delta_{jj} \sim \mathcal{N}(0,\epsilon^2)$.
\item $n \equiv d-r$ elements that are eigenvalues of $\delta_{\mathrm{ker}}$, which we denote by $\bvec{\kappa} = \{\kappa_j:~j = 1,\ldots, 
n\}$.
\end{enumerate}
In turn, $q$ is the solution to
\begin{equation}
\label{eq:q_eqn}
 \sum_{j=1}^{r}(p_j - q)^{+} + \sum_{j=1}^{n}{(\kappa_j-q)^+} = 1,
\end{equation}
where $(x)^{+} = \max(x, 0)$, and $\lambda_{\mathrm{kite}}$ is
\begin{equation}
\label{eq:llrs_kite}
\epsilon^{2}\lambda_{\mathrm{kite}} = \sum_{j=1}^{r}[\rho_{jj}- (p_j-q)^{+}]^2 + \sum_{j=1}^{n}\left[(\kappa_j-q)^+\right]^2.
\end{equation}
 
To solve Equation \eqref{eq:q_eqn}, and derive an approximation for \eqref{eq:llrs_kite}, we use the fact that we are interested in computing the \emph{average value} of $\lambda_{\mathrm{kite}}$, which justifies approximating the random variable $q$ by a closed-form, deterministic value. To do so, we need to understand the behavior of $\bvec{\kappa}$. Developing such an understanding, and a theory of its \emph{typical value}, is the subject of the next section.

\subsubsection{Approximating the eigenvalues of a GUE($n$) matrix}
We first observe that while the $\kappa_j$ are random variables, they are \emph{not} normally distributed.  Instead, because $\delta_{\mathrm{ker}}$ is proportional to a $\mathrm{GUE}(n)$ matrix, for $n\gg1$, the distribution of any eigenvalue $\kappa_{j}$
converges to a Wigner semicircle distribution \cite{Wigner1958}, given by $\mathrm{Pr}(\kappa) = \frac{2}{\pi R^{2}}\sqrt{R^{2}-\kappa^{2}}$ for $|\kappa| \leq R$, with $R = 2\epsilon\sqrt{n}$.  The eigenvalues are not independent; they tend to avoid collisions (``level avoidance'' \cite{Tao2013}), 
and typically form a surprisingly regular array over the support of the Wigner semicircle.  Since our goal is to compute $\expect{\lambda_{\mathrm{kite}}}$, we can capitalize on this behavior by replacing each random sample of $\bvec{\kappa}$ with a 
\emph{typical sample} given by its order statistics $\bar{\bvec{\kappa}}$.  These are the average values of the \emph{sorted} 
$\bvec{\kappa}$, so $\overline{\kappa}_j$ is the average value of the $j^{\mathrm{th}}$ largest value of $\bvec{\kappa}$.  Large random samples 
are usually well approximated (for many purposes) by their order statistics even when the elements of the sample are 
independent, and level avoidance makes the approximation even better. 

\begin{figure}
\includegraphics[width=\columnwidth]{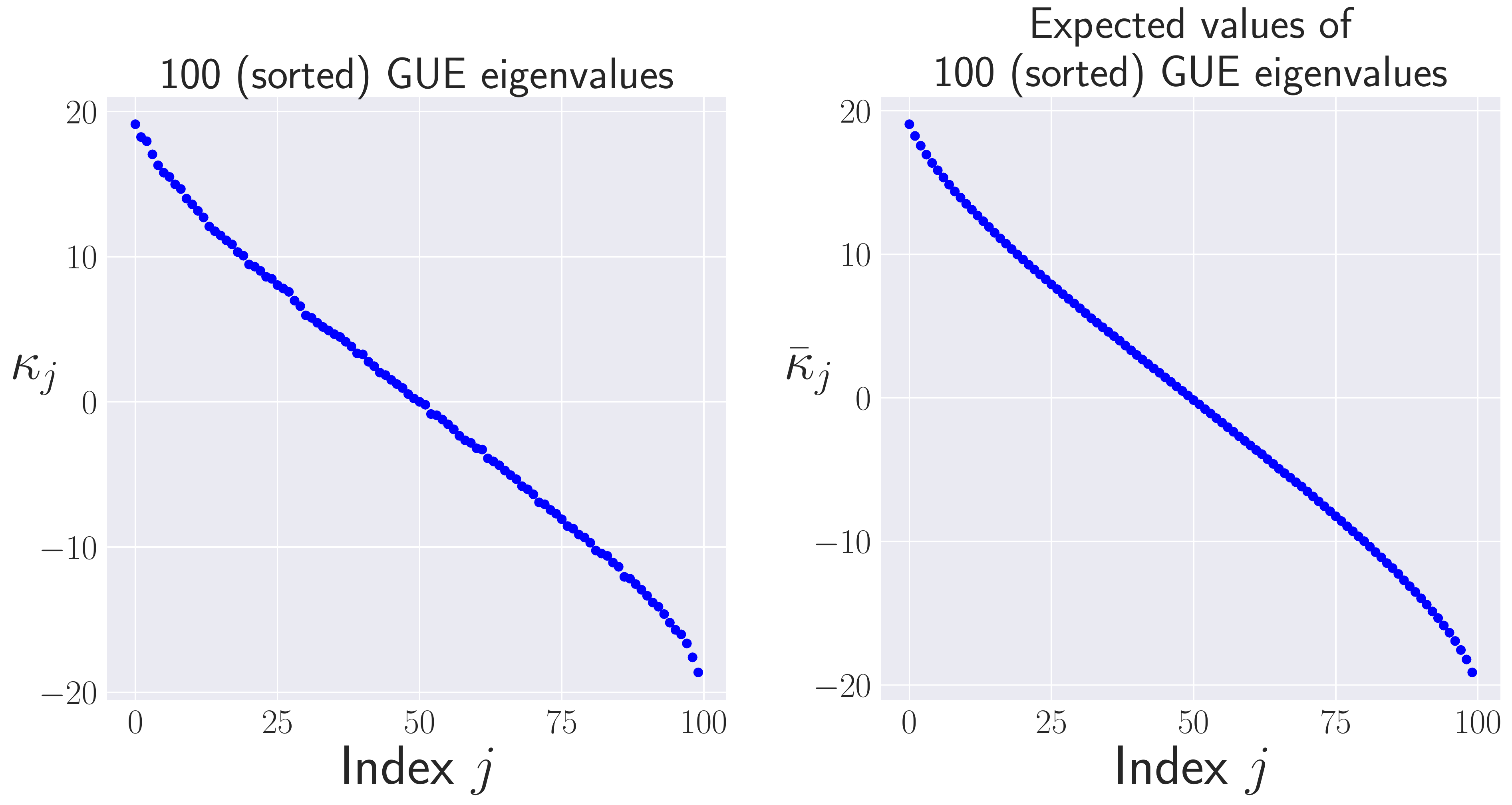}
\caption{\textbf{Approximating typical samples of GUE($n$) eigenvalues by order statistics:} We approximate a typical sample of GUE($n$) eigenvalues by their order statistics (average values of a sorted sample).  \textbf{Left}:  The sorted eigenvalues (i.e., order statistics $\kappa_{j}$) of one randomly chosen GUE(100) matrix.  \textbf{Right}:  Approximate expected values of the order statistics, $\bar{\kappa}_{j}$, of the GUE(100) distribution, computed as the average of the sorted eigenvalues of 100 randomly chosen GUE(100) matrices.}
\label{fig:orderstatistics1}
\end{figure}

Suppose that $\bvec{\kappa}$ are the eigenvalues of a GUE($n$) matrix, sorted from highest to lowest.  Figure \ref{fig:orderstatistics1} illustrates such a sample for $n=100$.  It also shows the \emph{average} values of 100 such samples (all sorted).  These are the \emph{order statistics} $\overline{\bvec{\kappa}}$ of the distribution (more precisely, what is shown is a good \emph{estimate} of the order statistics; the actual order statistics would be given by the average over infinitely many samples).  As the figure shows, while the order statistics \emph{are} slightly more smoothly and predictably distributed than a single (sorted) sample, the two are remarkably similar.  A single sample $\bvec{\kappa}$ will fluctuate around the order statistics, but these fluctuations are relatively small, partly because the sample is large, and partly because the GUE eigenvalues experience level repulsion.  Thus, the ``typical'' behavior of a sample -- by which we mean the mean value of a statistic of the sample -- is well captured by the order statistics (which have no fluctuations at all).

We now turn to the problem of modeling $\bvec{\kappa}$ quantitatively.  We note up front that we are only going to be interested in certain properties of $\bvec{\kappa}$:  specifically, partial sums of all $\kappa_j$ greater or less than the threshold $q$, or partial sums of functions of the $\kappa_j$ (e.g., $(\kappa_j-q)^2$).  We require only that an ansatz be accurate for such quantities.  We do not use this fact explicitly, but it motivates our approach -- and we do not claim that our ansatz is accurate for \emph{all} conceivable functions.

In general, if a sample $\bvec{\kappa}$ of size $n$ is drawn so that each $\kappa$ has the same probability density 
function $\mathrm{Pr}(\kappa)$, then a good approximation for the $j^{\mathrm{th}}$ order statistic is given by the inverse 
\emph{cumulative} distribution function (CDF):
\begin{equation}
\overline{\kappa}_j \approx \mathrm{CDF}^{-1}\left(\frac{j-1/2}{n}\right).
\end{equation}
This is closely related to the observation that the histogram of a sample tends to look similar to the underlying probability density function.  More precisely, it is equivalent to the observation that the empirical distribution function (the CDF of the histogram) tends to be (even more) similar to the underlying CDF.  For i.i.d. samples, this is the content of the Glivenko-Cantelli theorem \cite{VanderVaart2000}.  Figure \ref{fig:orderstatistics2} compares the order statistics of GUE(100) and GUE(10) eigenvalues (computed as numerical averages over 100 random samples) to the inverse CDF for the Wigner semicircle distribution.  Even though the Wigner semicircle model of GUE eigenvalues is only exact as $n\to\infty$, it provides a nearly-perfect model for $\overline{\bvec{\kappa}}$ even at $n=10$ (and remains surprisingly good all the way down to $n=2$).

\begin{figure}[t]
\includegraphics[width=\columnwidth]{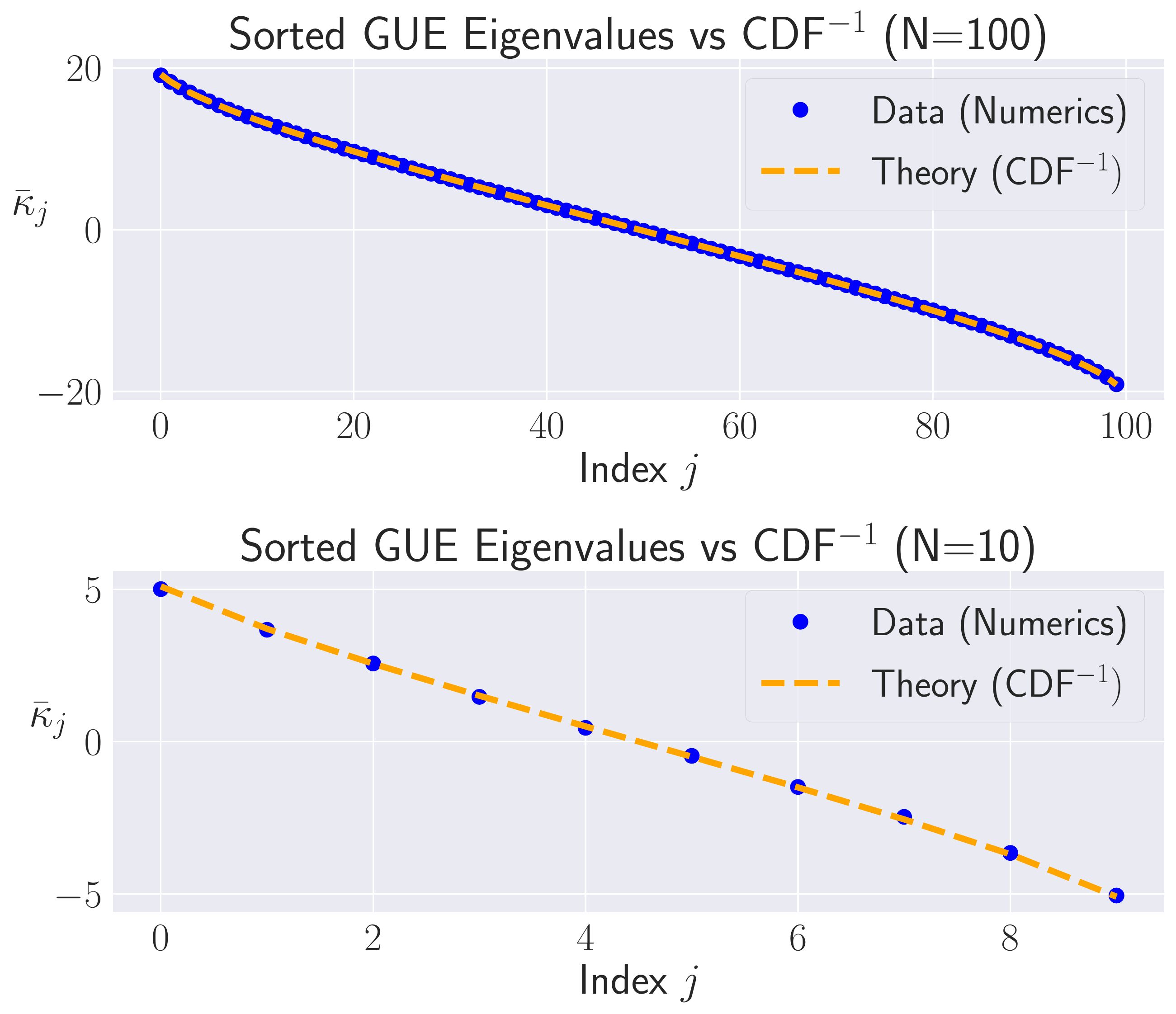}
\caption{\textbf{Approximating order statistics by the inverse CDF:} Order statistics of the GUE($n$) eigenvalue distribution are very well approximated by the inverse CDF of the Wigner semicircle distribution.  In both figures, we compare the order statistics of a GUE($n$) distribution to the inverse CDF of the Wigner semicircle distribution. \textbf{Top}:  $n=100$.  \textbf{Bottom}:  $n=10$.
Agreement in both cases is essentially perfect.}
\label{fig:orderstatistics2}
\end{figure}

We make one further approximation, by assuming that $n\gg1$, so the distribution of the $\overline{\kappa}_j$ is effectively continuous and identical to $\mathrm{Pr}(\kappa)$. For the quantities that we compute, this is equivalent to replacing the empirical distribution function (which is a step function) by the CDF of the Wigner semicircle distribution.  So, whereas for any given sample the partial sum of all $\kappa_j > q$ jumps discontinuously when $q=\kappa_j$ for any $j$, in this approximation it changes smoothly.  This accurately models the \emph{average} behavior of partial sums.

\subsubsection{Deriving an approximation for $q$}
The approximations of the previous section allow us to use $\{p_j\} \cup \{\overline{\kappa}_j\}$ as the ansatz for the eigenvalues of $\rhoML{\M'_{d}}$, where the $p_j$ are $\mathcal{N}(\rho_{jj},\epsilon^2)$ random variables, and the $\overline{\kappa}_j$ are the (fixed, smoothed) order statistics of a Wigner semicircle distribution.  In turn, the defining equation for $q$ (Equation \eqref{eq:q_eqn}) is well approximated as
\begin{equation}
\nonumber \sum_{j=1}^{r}(p_j - q)^{+} + \sum_{j=1}^{n}{(\overline{\kappa}_j-q)^+} = 1.
\end{equation}
To solve this equation, we observe that the $\overline{\kappa}_j$ are symmetrically distributed around $
\kappa=0$, so half of them are negative.  Therefore, with high probability, $\Tr
\left[\mathrm{Trunc}(\rhoML{\M'_{d}})\right]>1$, and so we will need to subtract $q\Id$ from $\rhoML{\M'_{d}}$ before truncating.

Because we have assumed $N_{\mathrm{samples}}$ is sufficiently large ($N_{\mathrm{samples}} >> \min_{j}1/\rho_{jj}^{2}$), the eigenvalues of $\rho_{0}$ are large compared to the perturbations $\delta_{jj}$ and $q$. This implies $(p_{j} - q)^{+} = p_{j} - q$. Under this assumption, $q$ is the solution to
\begin{align}
\nonumber \sum_{j=1}^{r}(p_j - q) + \sum_{j=1}^{n}{(\overline{\kappa}_j-q)^+} & = 1\\
\nonumber \implies - rq + \Delta + n\int_{\kappa=q}^{2\epsilon\sqrt{n}}{(\kappa-q)\mathrm{Pr}(\kappa)\mathrm{d}\kappa} & = 0\\
\nonumber \label{eq:q_eqn2}\implies - rq + \Delta + \frac{\epsilon}{12\pi}\left[(q^2+8n)\sqrt{-q^2+4n} \right.\\
\left. -12qn\left(\frac{\pi}{2}-\sin^{-1}\left(\frac{q}{2\sqrt{n}}\right)\right)   \right]&=0,
\end{align}
where $\Delta = \sum_{j=1}^{r}\delta_{jj}$ is a $\mathcal{N}(0,r\epsilon^2)$ random variable.  We choose to replace a discrete 
sum (line 1) with an integral (line 2). This approximation is valid when $n\gg1$, as we can accurately approximate a discrete collection of closely spaced real numbers by a smooth density or distribution over the real numbers that has approximately the same CDF.  It is also remarkably accurate in practice.
  
In yet another approximation, we replace $\Delta$ with its average value, which is zero.  We could obtain an even more accurate expression 
 by treating $\Delta$ more carefully, but this crude approximation turns out to be quite accurate already.

To solve Equation \eqref{eq:q_eqn2}, it is necessary to further simplify the complicated expression resulting from the integral (line 3).  To do so, we 
assume  $\rho_0$ is relatively low-rank, so $r \ll d/2$.  In this case, the sum of the positive $\overline{\kappa}_j$ is large compared 
with $r$, almost all of them need to be subtracted away, and therefore $q$ is close to $2\epsilon\sqrt{n}$.  We therefore replace 
the complicated expression with its leading order Taylor expansion around $q=2\epsilon\sqrt{n}$, substitute into Equation \eqref{eq:q_eqn2}, and 
obtain the equation
\begin{equation}
\frac{rq}{\epsilon}  = \frac{4}{15\pi}n^{1/4}\left(2\sqrt{n}-\frac{q}{\epsilon}\right)^{5/2}.
\end{equation}
This equation is a quintic polynomial in $q/\epsilon$, so by the Abel-Ruffini theorem, it has no algebraic solution.  However, as $n \rightarrow \infty$, its roots have a well-defined algebraic \emph{approximation} that becomes accurate quite rapidly (e.g., for $n>4$):
\begin{align}
\label{eq:truncation}
z &\equiv q/\epsilon \approx 2\sqrt{n}\left(1 -\frac{1}{2}x +\frac{1}{10}x^{2} - \frac{1}{200}x^{3}\right),
\end{align}
where $x = \left(\frac{15 \pi r}{2n}\right)^{2/5}$ \footnote{See supplemental information for a derivation of this solution.}.

\subsubsection{Expression for $\langle \lambda_{\mathrm{kite}}\rangle$}
Now that we know how much to subtract off in the truncation process, we can approximate $\expect{\lambda_{\mathrm{kite}}}$, originally given in Equation \eqref{eq:llrs_kite}:
\begin{align}
\nonumber \expect{\lambda_{\mathrm{kite}}} &\approx  \frac{1}{\epsilon^{2}}\left\langle\sum_{j=1}^{r}[\rho_{jj}- (p_j-q)^{+}]^2 + \sum_{j=1}^{n}\left[(\bar{\kappa}_j-q)^+\right]^2 \right\rangle\\
\nonumber &\approx \frac{1}{\epsilon^{2}} \left\langle\sum_{j=1}^{r}[-\delta_{jj} +  q ]^2 + \sum_{j=1}^{n}\left[(\bar{\kappa}_j-q)^+\right]^2 \right\rangle\\
\nonumber  &\approx r + rz^2 + \frac{n}{\epsilon^{2}}\int_{\kappa=q}^{2\epsilon\sqrt{n}}{ \mathrm{Pr}(\kappa)(\kappa-q)^2 d\kappa} \\
\nonumber &=r + rz^{2} + \frac{n(n+z^{2})}{\pi}\left(\frac{\pi}{2} - \sin^{-1}\left(\frac{z}{2\sqrt{n}}\right)\right) \\
& - \frac{z(z^{2}+26n)}{24\pi}\sqrt{4n-z^{2}}.
\end{align}

\subsection{Complete Expression for $\langle \lambda \rangle$}
\label{subsec:LLRS}
The total expected value, $\expect{\lambda} = \expect{\lambda_{\mathrm{L}}} + \expect{\lambda_{\mathrm{kite}}}$, is thus
\begin{align}
\label{eq:ourLLRS}
\nonumber \langle \lambda(\rho_{0}, \M_{d}) \rangle &\approx 2rd - r^{2}+rz^{2}\\
\nonumber & + \frac{n(n+z^{2})}{\pi}\left(\frac{\pi}{2} - \sin^{-1}\left(\frac{z}{2\sqrt{n}}\right)\right) \\
& - \frac{z(z^{2}+26n)}{24\pi}\sqrt{4n-z^{2}}.
\end{align}
where $z$ is given in Equation \eqref{eq:truncation}, $n=d-r$, and $r = \mathrm{Rank}(\rho_{0})$.

This null theory is much more complicated than the Wilks theorem, but as Figure \ref{fig:modelcomp-iso} shows, it is very accurate when $2r << d$. (In contrast, the prediction of the Wilks theorem is wildly incorrect for $r\ll d$.) Although our null theory does break down as $r \rightarrow d$, it does so fairly gracefully. We conclude that our analysis (and Equation \eqref{eq:ourLLRS}) correctly models tomography \emph{if} the Fisher information is isotropic ($\Fi \propto \Id$), a point we turn to in the following subsections.

\begin{figure}
 \includegraphics[width=\columnwidth]{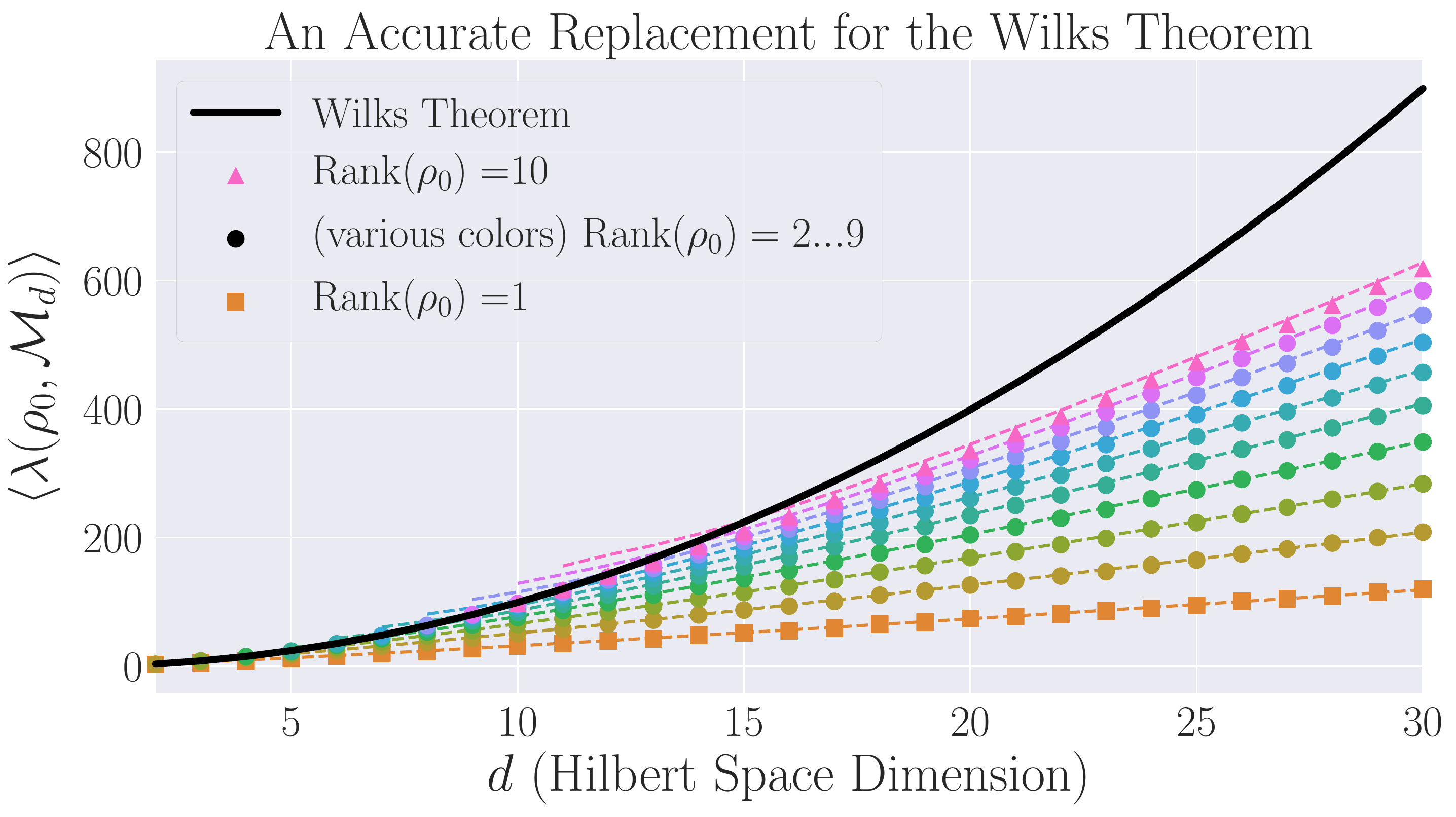}
 \caption{\textbf{Improved prediction for $\langle \lambda(\rho_{0}, \M_{d})\rangle$, as compared to the Wilks theorem:} Numerical results for $\expect{\lambda(\rho_{0}, \M_{d})}$ compared to the prediction of the Wilks theorem (solid line) and our replacement theory as given in Equation \eqref{eq:ourLLRS} (dashed lines).  Our formula depends on the rank $r$ of $\rho_0$ (unlike the Wilks prediction), and is nearly perfect for $r\ll d/2$.  It becomes less accurate as $r$ approaches $d/2$, and is invalid when $r\approx d$.}
 \label{fig:modelcomp-iso}
\end{figure}

In the asymptotic limit $d\rightarrow \infty$, while keeping $r$ fixed, $\langle \lambda \rangle$ takes the following form:
\begin{equation}
\langle \lambda \rangle \underset{d\rightarrow\infty}{\longrightarrow}~rd\left[6 - \frac{20}{7}\left(\frac{15\pi r}{2d}\right)^{2/5} +\frac{20}{21}\left(\frac{15\pi r}{2d}\right)^{4/5}\right]-5r^{2}.
\end{equation}
That $\langle \lambda \rangle$ scales as $\mathcal{O}(rd)$ in this regime is to be expected, as a rank-$r$ density matrix has $\mathcal{O}(rd)$ free parameters. Curiously though, this asymptotic result is not equal to $\langle \lambda_{\mathrm{L}}\rangle$, meaning that the ``kite" elements continue to contribute to the behavior of the statistic (even though most of them will be set to zero in the projection step when computing $\rhoML{\M}$).

\section{Comparison to Numerical Experiments}

To evaluate our null theory for $\langle \lambda \rangle$, we compare it to numerical experiments, described below.

\subsection{Comparison to Idealized Tomography (Isotropic Fisher Information)}
\label{sec:theorycomp1}

In our derivation of Equation \eqref{eq:ourLLRS}, we assumed both that the Fisher information is isotropic and that the number of samples is asymptotically infinite. For a variety of true states $\rho_{0}$ with dimension  $d=2,\ldots,30$ and rank $r=1,\ldots,10$ we: (a) generated $N=500$ i.i.d. $\mathcal{N}(\rho_{0}, \epsilon^{2}\mathcal{I})$ unconstrained ML estimates $\{\rhoML{\M'_{d, j}}\}_{j=1}^{N}$, thereby simulating the unconstrained ML estimates \emph{at} the $N_{\mathrm{samples}} = \infty$ limit, (b) numerically solved Equation \eqref{eq:MP-LANmle2} for each $\rhoML{\M'_{d, j}}$ to obtain the constrained ML estimate $\rhoML{\M_{d, j}}$, and (c) estimated $\langle \lambda \rangle$ as $\frac{1}{N}\sum_{j=1}^{N}\mathrm{Tr}[(\rho_{0} - \rhoML{\M_{d,j}})^{2}]/\epsilon^{2}$. We took $\epsilon = 10^{-4}$, to ensure that all of the unconstrained ML estimates are close to $\rho_{0}$, and that we are not erroneously generating estimates which are too far away. (Recall that our derivation used the fact that, asymptotically, we can ``zoom in" on $\rho_{0}$ to understand the behavior of $\lambda$. Consequently, if $\epsilon$ is too large, then some of the unconstrained ML estimates may almost be orthogonal to $\rho_{0}$, which clearly violates the conditions used in our derivation.)

In Figure \ref{fig:modelcomp-iso}, we compare our theory (dashed lines) to these numerical results (solid dots).  It is clear Equation \eqref{eq:ourLLRS} is almost perfectly accurate when $r \ll d/2$, but it does begin to break down as $r$ becomes comparable to $d$.

\subsection{Comparison to Heterodyne Tomography}
\label{sec:heterotomo}
In practice, the Fisher information is rarely isotropic.  So we tested our idealized result by applying it to a realistic, challenging, and experimentally relevant problem: quantum heterodyne (equivalent to double homodyne) state tomography \cite{Lvovsky2001a, Bertrand1987, Leonhardt1995, Lvovsky2009} of a single optical mode.  (See Figure \ref{fig:fish_condition} for a plot of the \emph{condition number} -- the ratio of the largest eigenvalue to the smallest -- of the estimated Fisher information. It is clear $\mathcal{I} \not \propto \Id$.) States of this continuous-variable system are described by density operators on the infinite-dimensional Hilbert space $L^2(\reals)$.  Fitting these infinitely many parameters to finitely much data demands simpler models.

\begin{figure}
  \includegraphics[width=.9\columnwidth]{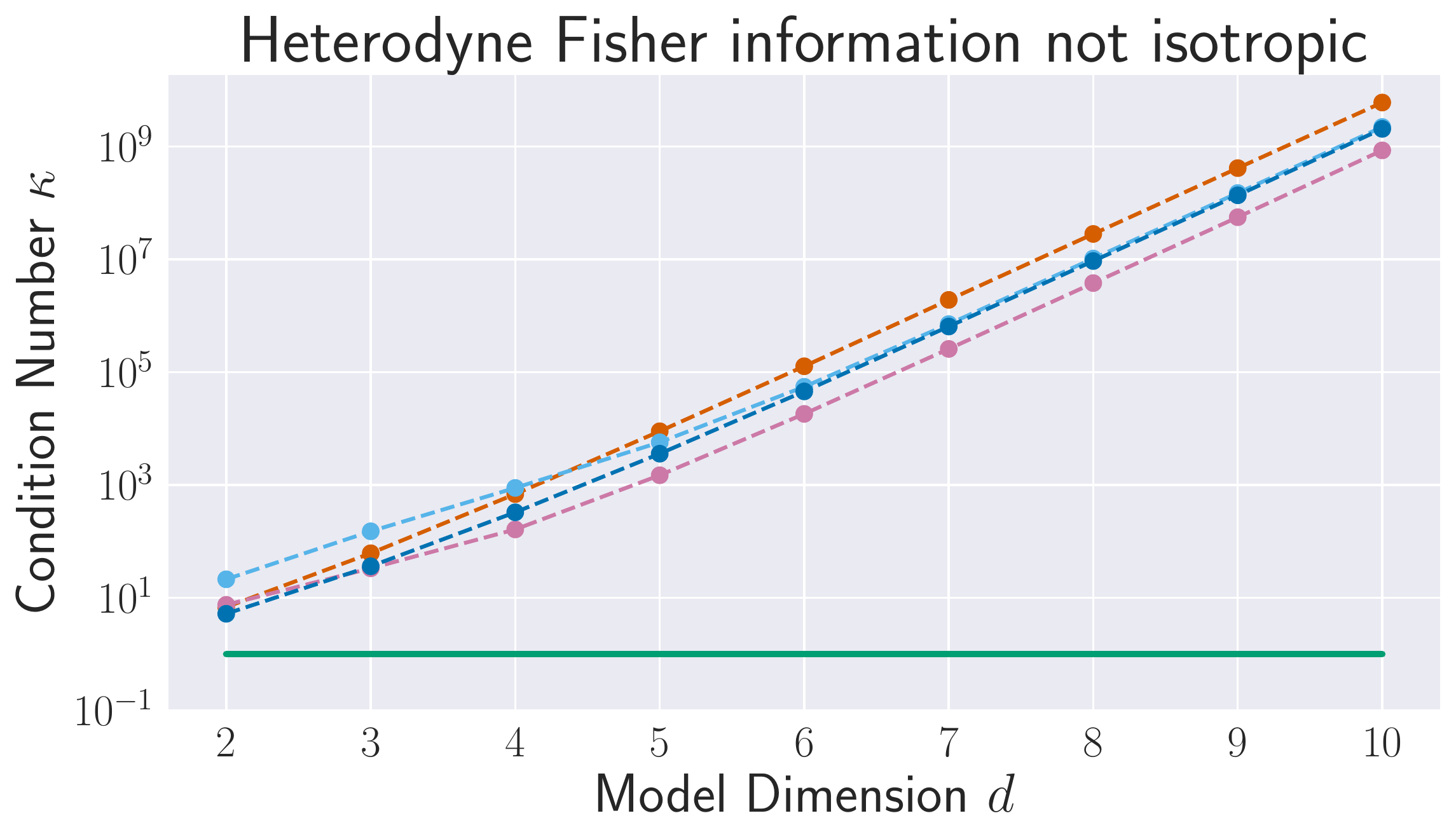}
 \caption{\textbf{Anisotropy of the heterodyne POVM Fisher information:} The condition number $\kappa$ -- the ratio of the largest eigenvalue to the smallest -- of the estimated heterodyne Fisher information. (Estimates are the average over 100 Hessians of the loglikelihood function.) $\kappa$ grows with model dimension, meaning anisotropy is increasing.  The dashed lines indicate different states $\rho_{0}$, and the solid line is $\kappa = 1$ (i.e., $\mathcal{I} \propto \Id$.).}
\label{fig:fish_condition}
\end{figure}

We consider a family of nested models motivated by a low-energy (few-photon) ansatz, and choose   
the Hilbert space $\mathcal{H}_d$ to be that spanned by the photon number states $\{\ket{0},\ldots ,\ket{d-1}\}$.
Heterodyne tomography reconstructs $\rho_{0}$ using data from repeated measurements of the 
coherent state POVM, $\{|\alpha\rangle\langle \alpha| /\pi, ~\alpha=x+ip\in \mathbb{C}\}$, which corresponds to sampling directly from the Husimi $Q$-function of $\rho_{0}$ \cite{Husimi1940}.

We examined the behavior of $\lambda$ for 13 distinct states $\rho_{0}$, both pure and mixed, supported on $\mathcal{H}_{2}, \mathcal{H}_{3}, \mathcal{H}
_{4}$, and $\mathcal{H}_{5}$.  We used rejection sampling to simulate 100 heterodyne datasets with up to $N_{\mathrm{samples}}=10^5$, and found ML estimates $\rhoML{\M_{d}}$ over each of the 9 models $\M_2, \ldots, M_{10}$ using numerical optimization \footnote{The model $\M_{1}$ is trivial, as $\M_{1} = \{|0\rangle \langle 0|\}$. This model will almost always be wrong, in general.}.  For each $\rho_{0}$ and each $d$, we averaged $\lambda(\rho_{0}, \M_{d})$ over all 100 datasets to obtain an empirical average loglikelihood ratio $\bar{\lambda}(\rho_0,d)$.

\begin{figure}
 \includegraphics[width=.9\columnwidth]{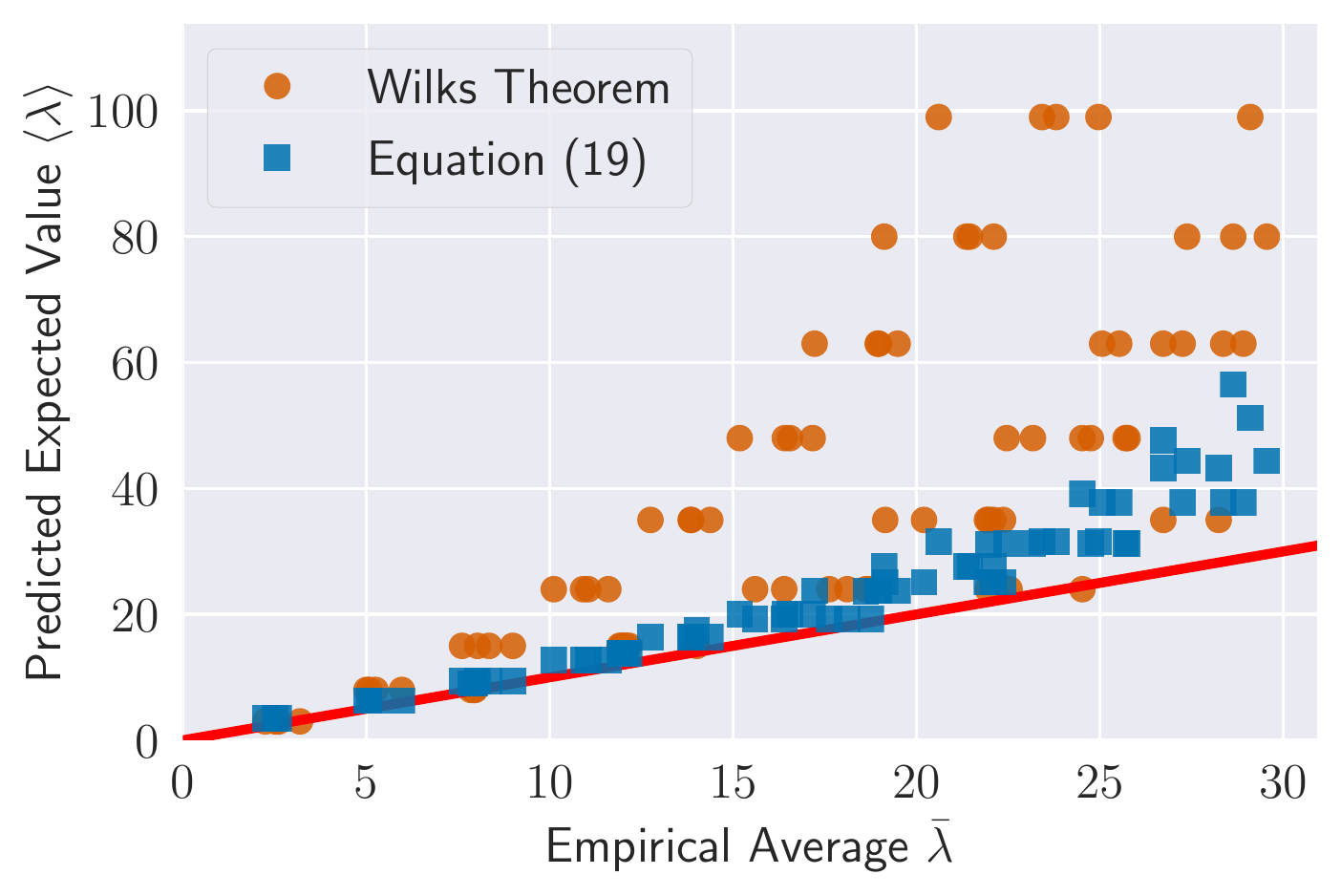}
 \caption{\textbf{Applying isotropic formula to heterodyne tomography:} The Wilks theorem (orange dots) dramatically over-estimates $\langle\lambda(\rho_{0}, \M_{d})\rangle$ in optical heterodyne tomography. Our formula, Equation \ref{eq:ourLLRS} (blue squares), is far more accurate. Residual discrepancies occur in large part because $N_{\mathrm{samples}}$ is not yet ``asymptotically large''. The solid red line corresponds to perfect correlation between theory ($\expect{\lambda}$) and practice ($\bar\lambda$).}
 \label{fig:modelcomp}
\end{figure}

Results of this test are shown in Figure \ref{fig:modelcomp}, where we plot the predictions for $\langle \lambda \rangle$ given by the Wilks theorem and Equation \eqref{eq:ourLLRS}, against the empirical average $\bar\lambda$, for a variety of $\rho_{0}$ and $d$. Our formula correlates very well with the empirical average, while the Wilks theorem (unsurprisingly) overestimates $\lambda$ dramatically for low-rank states.  Whereas a model selection procedure based on the Wilks theorem would tend to falsely reject larger Hilbert spaces (by setting the threshold for acceptance too high), our formula provides a reliable null theory.

Interestingly, as $d$ grows, Equation \eqref{eq:ourLLRS} also begins to overpredict. As Figure \ref{fig:totalcontrib} indicates, a more accurate description is that the numerical experiments are \emph{underachieving}, because $\bar\lambda$ is still growing with $N_{\mathrm{samples}}$.  Both the Wilks theorem and our analysis are derived in the limit $N_{\mathrm{samples}} \rightarrow \infty$; for finite but large $N_{\mathrm{samples}}$, both may be invalid.  Figure \ref{fig:totalcontrib} shows that, even at $N_{\mathrm{samples}}\sim 10^{5}$, the behavior of $\bar{\lambda}$ has failed to become asymptotic. This is surprising, and suggests heterodyne tomography is a particularly exceptional and challenging case to model statistically. 

\begin{figure}
  \includegraphics[width=.9\columnwidth]{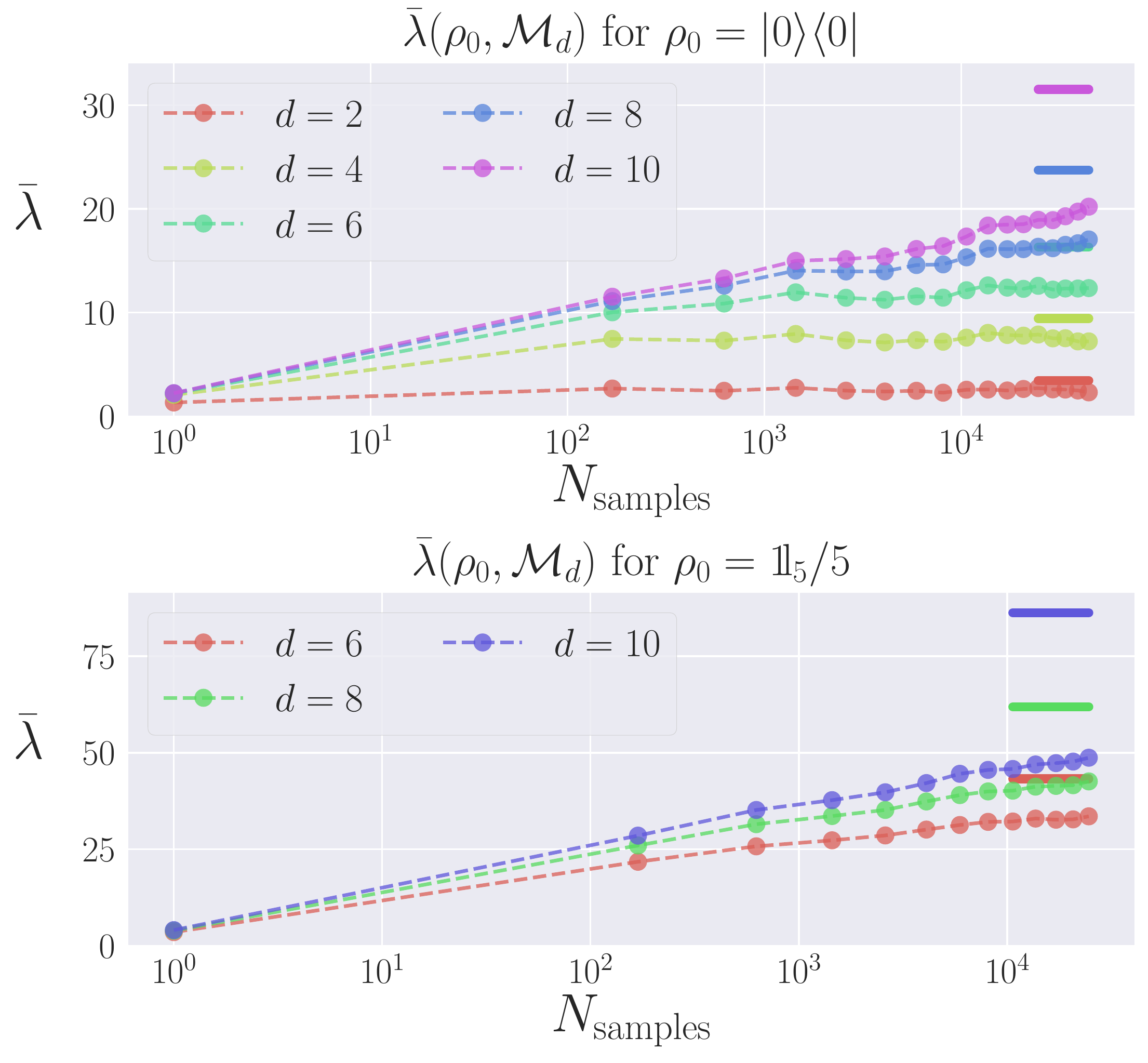}
 \caption{\textbf{``Underachievement" of $\bar{\lambda}$ in heterodyne tomography:} The empirical average $\bar{\lambda}$  may have achieved its asymptotic value, or is still 
growing, depending on $\rho_{0}$ and $d$. Solid lines indicate the value of our formula
for the asymptotic expected value, given in Equation \eqref{eq:ourLLRS}.}
\label{fig:totalcontrib}
\end{figure}

\begin{figure}
  \includegraphics[width=\columnwidth]{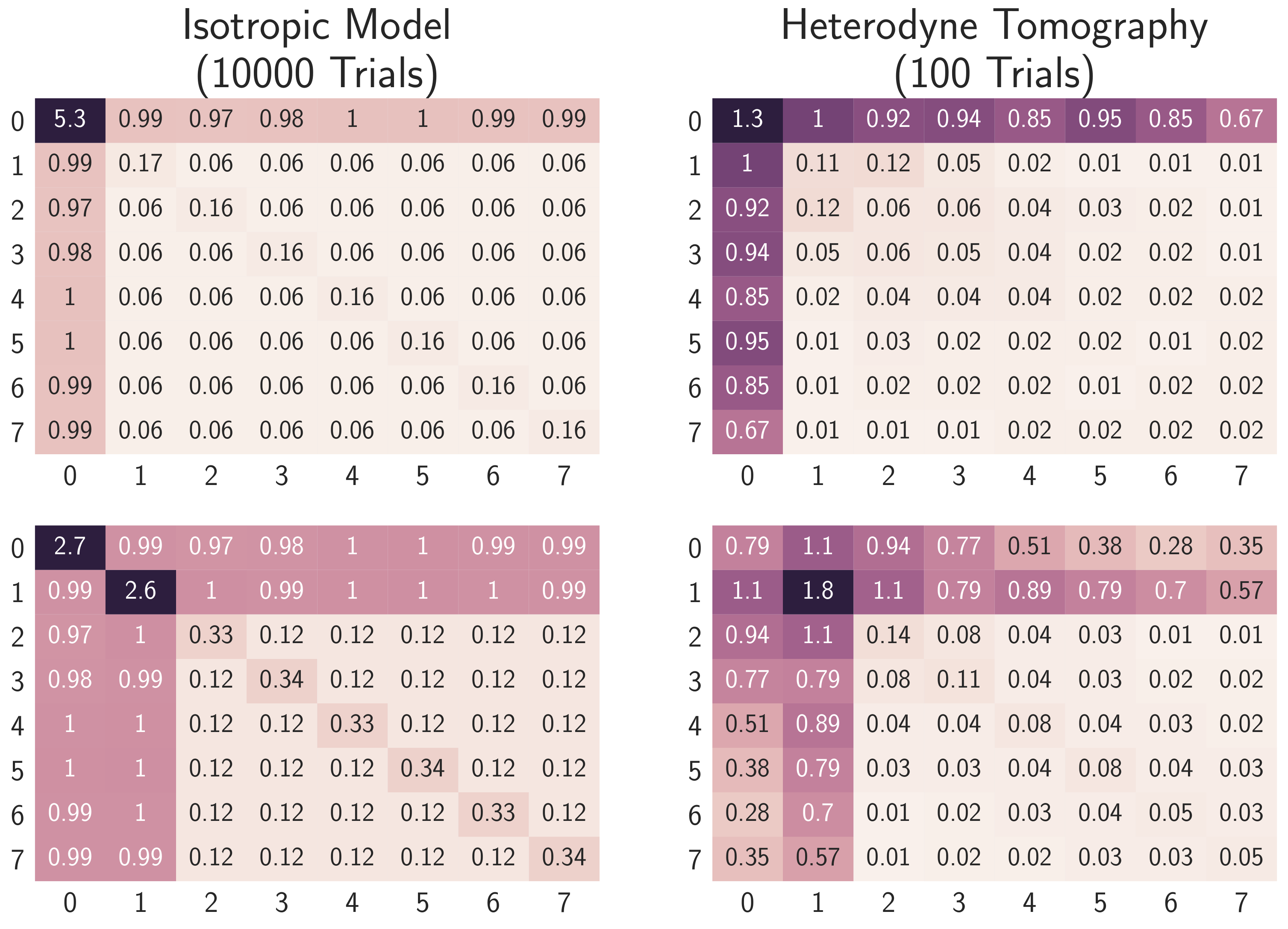}
 \caption{\textbf{Detailed comparison of isotropic model and heterodyne tomography:} The values of $\langle \lambda_{jk} \rangle$ for an isotropic Fisher information (left), and for heterodyne tomography (right). \textbf{Top}: $\rho_{0} = |0\rangle\langle 0|$. \textbf{Bottom}: $\rho_{0} = \mathcal{I}_{2}/2$. \textbf{Discussion}: Qualitatively, the behavior is the same, though there are quantitative differences, particularly within the kite.}
\label{fig:model_comparison}
\end{figure}

However, our model \emph{does} get some of the qualitative features correct. In Figure \ref{fig:model_comparison}, we present simulated values of $\langle \lambda_{jk}\rangle$ for an isotropic Fisher information and for heterodyne tomography. While the values of $\langle \lambda_{jk} \rangle$ do not agree exactly, they still decompose into two types, the ``L" and the ``kite". (See Figure \ref{fig:individcontrib} for an analysis of the discrepancies.)

\begin{figure}
  \includegraphics[width=\columnwidth]{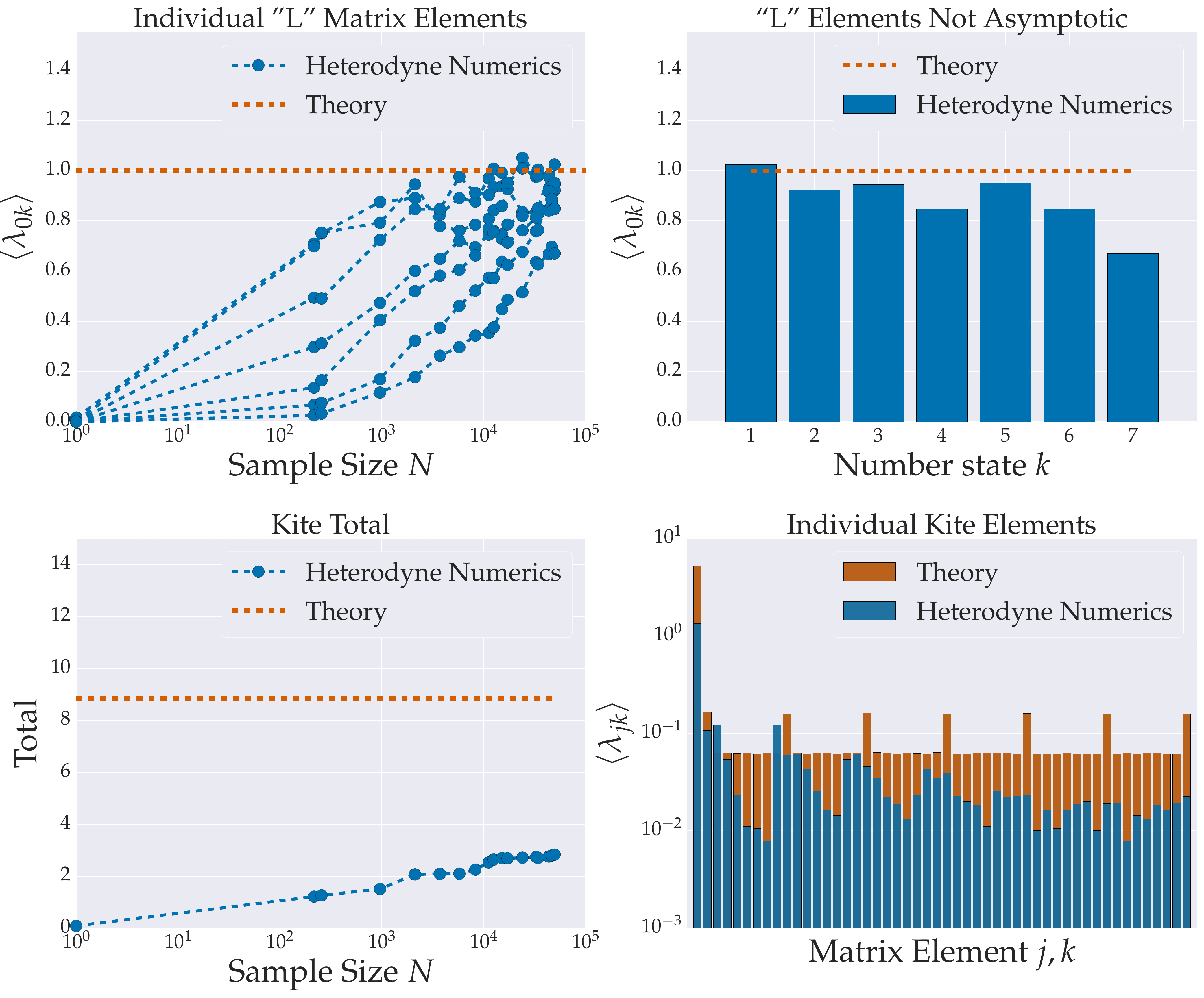}
 \caption{\textbf{Discrepancies between isotropic model and heterodyne tomography:} Examining how our prediction for $\langle \lambda_{jk} \rangle$ disagrees with simulated heterodyne experiments. We take $\rho_{0} = |0\rangle\langle 0|$ and $d=8$. \textbf{Top Left}: The values of  $\langle \lambda_{0k}\rangle$ in the ``L" as a function of $N_{\mathrm{samples}}$.  \textbf{Top Right}:  At the largest $N_{\mathrm{samples}}$ studied, $\langle \lambda_{0k}\rangle$ is less than 1, especially for the higher number states. \textbf{Bottom Left}: The total from the ``kite" versus $N_{\mathrm{samples}}$. It is clear the total is still growing. \textbf{Bottom Right}: The individual ``kite" elements $\langle \lambda_{jk}\rangle$ at the largest $N_{\mathrm{samples}}$ studied;  most are small compared to their values in the isotropic case.}
\label{fig:individcontrib}
\end{figure}

~\\
\section{Conclusions and Discussion}
Quantum state space violates local asymptotic normality, a key property satisfied by classical statistical models. Through the introduction of metric-projected local asymptotic normality, we have provided a new framework for reasoning about classical statistical results for models that don't satisfy LAN because of convex constraints.

We explicitly investigated one such result, the Wilks theorem, found it is not generally reliable in quantum state tomography, and provided  a much more broadly applicable replacement that can be used in model selection methods (Equation \eqref{eq:ourLLRS}).  This includes information criteria such as the AIC and BIC \cite{Akaike1974, Schwarz1978, Kass1995, Burnham2004} that do not explicitly use the Wilks theorem, but rely on the same assumptions (local asymptotic normality, equivalence between loglikelihood and squared error, etc.).  Null theories of loglikelihood ratios have many other applications, including hypothesis testing \cite{Blume-Kohout2010,Moroder2013} and confidence regions \cite{Glancy2012a}, and our result is directly applicable to them.  Refs. \cite{Moroder2013,Glancy2012a} both point out explicitly that their methods are unreliable near boundaries and therefore cannot be applied to rank-deficient states; our result fixes this outstanding problem.

However, our numerical experiments with heterodyne tomography show unexpected behavior, indicating that quantum tomography can still surprise, and may violate \emph{all} asymptotic statistics results.  In such cases, bootstrapping \cite{Efron1979, Higgins2004} may be the only reliable way to construct null theories for $\lambda$. 

Finally, the \emph{methods} presented here have application beyond the analysis of loglikelihoods.  Metric-projected local asymptotic normality provides a new and rigorous framework for describing the behavior of the maximum likelihood estimator in the presence of convex constraints. This framework helps shed light on the behavior of $\rhoML{d}$ for rank-deficient states, and can be used to predict other derived properties such as the average rank of the estimate, which is independently interesting for (e.g.) quantum compressed sensing \cite{Flammia2012a, Steffens2016, Kalev2015, Kalev2015a}.

\section{Acknowledgements:}
TLS thanks Jonathan A Gross for helpful discussions on software design, coding, and statistics, John King Gamble for useful insights on parallelized
computation and feedback on earlier drafts of this paper, and Anupam Mitra, Kenneth M Rudinger, and Daniel Suess for proofreading edits. The authors are grateful for those who provide support for the following software packages: iPython/Jupyter \cite{Perez}, matplotlib
\cite{Hunter2007}, mpi4py \cite{Dalcin2011},  NumPy \cite{VanDerWalt2011}, pandas \cite{mckinney2010}, Python 2.7 
\cite{vanRossum}, SciPy \cite{Oliphant2007a}, seaborn \cite{Waskom2016}, and SymPy \cite{Meurer2017}. 

Sandia National Laboratories is a multimission laboratory managed and operated by National Technology \& Engineering Solutions of Sandia, LLC, a wholly owned subsidiary of Honeywell International Inc., for the U.S. Department of Energy's National Nuclear Security Administration under contract DE-NA0003525.

\bibliographystyle{apsrev4-1}
\bibliography{wilkspaper}
\end{document}